\documentclass[review]{elsarticle}
\usepackage[letterpaper,top=2cm,bottom=2cm,left=2.9cm,right=2.9cm,marginparwidth=1.75cm]{geometry}
\usepackage{lineno,hyperref,mathtools}

\journal{Journal of \LaTeX\ Templates}

\usepackage{amssymb}
\usepackage{bm}
\usepackage{amssymb}
\usepackage{amsthm}
\usepackage{multirow,booktabs}
\usepackage{cases}
\usepackage{threeparttable}
\usepackage{color}
\usepackage{rotating}
\newtheorem{theorem}{Theorem}

\newtheorem{remark}{Remark}
\newtheorem{definition}{Definition}
\newtheorem{lemma}{Lemma}
\newtheorem{corollary}{Corollary}
\newtheorem{example}{Example}

\newcommand{\Hull}{{\mathrm{Hull}}}
\newcommand{\C}{{\mathcal{C}}}
\newcommand{\F}{{\mathbb{F}}}

\begin{document}

\begin{frontmatter}

\title{Galois self-orthogonal MDS codes with large dimensions}
\tnotetext[mytitlenote]{This research work is supported by the National Natural Science Foundation of China under Grant Nos. U21A20428 and 12171134.}

\author[mymainaddress]{Ruhao Wan}
\ead{wanruhao98@163.com}

\author[mymainaddress]{Shixin Zhu\corref{mycorrespondingauthor}}
\cortext[mycorrespondingauthor]{Corresponding author}
\ead{zhushixinmath@hfut.edu.cn}

\address[mymainaddress]{School of Mathematics, HeFei University of Technology, Hefei 230601, China}

\begin{abstract}
Let $q=p^m$ be a prime power, $e$ be an integer with $0\leq e\leq m-1$ and $s=\gcd(e,m)$.
In this paper, for a vector $\bm{v}\in (\F_q^*)^n$ and a $q$-ary linear code $\C$,
we give some necessary and sufficient conditions for the equivalent code $\Phi_{\bm{v}}(\C)$ of $\C$ and the extended code of $\Phi_{\bm{v}}(\C)$ to be $e$-Galois self-orthogonal.
From this, we directly obtain some necessary and sufficient conditions for (extended) generalized Reed-Solomon (GRS and EGRS) codes to be $e$-Galois self-orthogonal.
Furthermore, for all possible $e$ satisfying $0\leq e\leq m-1$,
we classify them into three cases (1) $\frac{m}{s}$ odd and $p$ even;
 (2) $\frac{m}{s}$ odd and $p$ odd; (3) $\frac{m}{s}$ even, and construct several new classes of $e$-Galois self-orthogonal maximum distance separable (MDS) codes.
 It is worth noting that our $e$-Galois self-orthogonal MDS codes can have dimensions greater than $\lfloor \frac{n+p^e-1}{p^e+1}\rfloor$, which are not covered by previously known ones.
Moreover, by propagation rules, we obtain some new MDS codes
with Galois hulls of arbitrary dimensions.
 As an application, many quantum codes can be obtained from these MDS codes with Galois
hulls.
\end{abstract}


\end{frontmatter}

\section{Introduction}\label{sec1}

Let $\F_{q}$ be the finite field with $q$ elements, where $q=p^m$ is a prime power.
An $[n,k,d]_{q}$ linear code $\C$ is a $k$-dimensional subspace of $\F_{q}^{n}$ with minimum distance $d$.
It is well known that for a linear code $\C$, it needs to satisfy the Singleton bound: $d\leq n-k+1$.
In particular, if a linear code $\C$ can reach the Singleton bound (i.e., $d=n-k+1$), then $\C$ is called an MDS code.
Since MDS codes have good error correction capability,
the construction of MDS codes has received a lot of attention.
In particular, GRS and EGRS codes are the most well-known family of MDS codes.

Let $e$ be an integer with $0\leq e\leq m-1$ and $s=\gcd(e,m)$.
In \cite{RefJ (2017) Galois},
 Fan et al. first introduced the $e$-Galois inner product.
 The Galois inner product is a generalization of the Euclidean and Hermitian inner product.
 The $e$-Galois hull of a linear code $\C$ is defined as $\Hull_e(\C)=\C\cap \C^{\bot_e}$,
 where $\C^{\bot_e}$ is the $e$-Galois dual code of $\C$.
 If $\C\subseteq \C^{\perp_e}$ (i.e., $\Hull_e(\C)=\C$), the linear code $\C$ is called $e$-Galois self-orthogonal.
 In particular, if $\C=\C^{\perp_e}$ (i.e., $\Hull_e(\C)=\C=\C^{\bot_e}$),  $\C$ is called $e$-Galois self-dual.

 On one hand, due to the nice algebraic structure of Euclidean self-dual (resp. self-orthogonal) MDS codes,
 the construction of Euclidean self-dual (resp. self-orthogonal) MDS codes has been extensively studied (see \cite{RefJ (2020) Feng MDS it}-\cite{RefJ (2021) Fang It} and the references therein).

On the other hand, the authors in \cite{RefJ (1998) introduct 1}-\cite{RefJ (2001) introduct 3} established the fundamentals for constructing a class of quantum codes named stabilizer codes using linear codes with certain self-orthogonality.
Since then, Hermitian self-orthogonal MDS codes are widely used to construct MDS quantum error-correcting codes (QECCs) (see \cite{RefJ (2008) Li},\cite{RefJ (2018) Fang two},\cite{RefJ (2023) Wan 2},\cite{RefJ (2004) M. R}-\cite{RefJ (2004) M. G} and the references therein).

In \cite{RefJ (2006) bound}, Brun et al. introduced a class of quantum codes called entanglement-assisted
quantum error-correcting codes (EAQECCs) which contains quantum stabilizer codes as a special case.
Since then, the Euclidean (resp. Hermitian or Galois) hulls of linear codes has been extensively studied and used to construct (MDS) EAQECCs (see \cite{RefJ (2019) Luo}-\cite{RefJ (2022) Luo rule},\cite{RefJ (2020) Fang hull}-\cite{RefJ (2024) relative hull},\cite{RefJ (2023) Yang DM},\cite{RefJ (2023) Wan 3},\cite{RefJ (2018) K good},\cite{RefJ (2023) Chen},\cite{RefJ (2018) Liu. Galois},\cite{RefJ (2019) Liu  FFA} and the references therein).

\subsection{Related works}

In recent years,
researchers have made efforts to construct Galois self-dual (resp. self-orthogonal) codes
via skew multi-twisted codes (see \cite{RefJ (2019) Sharma FFa}) and constacyclic codes (see \cite{RefJ (2017) Galois}-\cite{RefJ (2021) Mi DM}).
Moreover, in \cite{RefJ (2020) liu hull}-\cite{RefJ (2022) Wu FFa}, the Galois hulls of linear codes was extensively studied.
Recently, (extended) GRS codes (resp. Gabidulin codes), an important class of MDS codes, are also widely used to construct MDS codes with $e$-Galois hulls of arbitrary dimensions.
We summarise the relevant results as follows.

\begin{itemize}
\item
In \cite{RefJ (2021) cao}-\cite{RefJ (2023) cao},
Cao constructed eleven classes of MDS codes with Galois hulls of arbitrary dimension.
Among them, we can find four classes of Galois self-orthogonal MDS codes.

\item
In \cite{RefJ (2022) Fang CC},
Fang et al. gave a class of Galois self-orthogonal MDS codes obtained from Euclidean self-orthogonal GRS codes.

\item
In \cite{RefJ (2023) Yang CC}-\cite{RefJ (2023) Yang FFA},
Li et al. constructed eight classes of Galois self-orthogonal MDS codes.
Moreover, the authors gave a class of $e$-Galois self-orthogonal MDS codes obtained from $s$-Galois self-orthogonal GRS codes,
where $2s\mid m$.

\item
In \cite{RefJ (2024) Qian dcc},
Qian et al. constructed nine classes of MDS codes with Galois hulls of arbitrary dimension.
Among them, we can find six classes of Galois self-orthogonal MDS codes.

\item
In \cite{RefJ (2023) Islam}, Islam et al. studied the Galois hull dimensions of Gabidulin codes.
This leads to a class of Galois self-orthogonal MDS codes with parameters $[m,k]_{q}$ for any $1\leq k\leq \min\{e,m-e\}$.
\end{itemize}
We summarize the currently known Galois self-orthogonal MDS codes in Table \ref{tab:1} of Section \ref{sec5}.

\subsection{Our motivation}

Our main motivations can be summarized as follows:
\begin{itemize}
\item
Recently,
with the aims of constructing
EAQECCs and MDS codes with Galois hulls, GRS codes have been widely studied under the Galois inner product
(see \cite{RefJ (2021) cao}-\cite{RefJ (2023) cao},\cite{RefJ (2022) Fang CC}-\cite{RefJ (2023) Yang FFA},\cite{RefJ (2024) Qian dcc},\cite{RefJ (2023) Liu Galois}-\cite{RefJ (2022) Wu FFa}).
Constructing Galois self-orthogonal MDS codes with large dimensions can be used to obtain new MDS codes with Galois hulls of arbitrary dimensions.

\item
Compared to the Euclidean and Hermitian inner product, the Galois inner product
has the more general setting, which allows us to to find more codes with better algebraic structures or good parameters.
Moreover, the study of Galois self-orthogonal GRS codes is a generalisation of
many of the current results for Euclidean (resp. Hermitian) self-orthogonal GRS codes.

\item
From Table \ref{tab:1}, we can find that
most of the known Galois self-orthogonal MDS codes have dimensions less than  $\lfloor \frac{n+p^e-1}{p^e+1}\rfloor$, where $2e\leq m$.
Therefore, as mentioned in \cite{RefJ (2022) Fang CC},
constructing new $e$-Galois self-orthogonal MDS codes with dimensions greater than $\lfloor \frac{n+p^e-1}{p^e+1}\rfloor$ is a challenging work.
\end{itemize}

\newcommand{\tabincell}[2]{\begin{tabular}{@{}#1@{}}#2\end{tabular}}
\begin{table}
\caption{Some known $e$-Galois self-orthogonal MDS codes}
\label{tab:1}
\begin{center}
\resizebox{\textwidth}{102mm}{
	\begin{tabular}{ccccc}
		\hline
		 Finite field $q$ & Length $n$ & Dimension $k$ & References \\
        \hline
         $q=p^m$ odd, $2e\mid m$ & $n=t\frac{q-1}{p^e-1}$, $1\leq t\leq p^e-1$ & $1\leq k\leq \lfloor \frac{n+p^e-\frac{q-1}{p^e-1}}{p^e+1}\rfloor$ & \cite{RefJ (2023) Yang FFA}  \\
        \hline
         $q=p^m$ odd, $2e\mid m$ & $n=t\frac{q-1}{p^e-1}+1$, $1\leq t\leq p^e-1$ &
         \tabincell{c}{$1\leq k\leq
        \max\{\lfloor \frac{n+p^e-1}{p^e+1}\rfloor,$\\$\lfloor \frac{n+p^{m-e}-1}{p^{m-e}+1}\rfloor\}$} & \cite{RefJ (2021) cao,RefJ (2023) cao}\\
        \hline
        $q=p^m$ odd, $2e\mid m$ & $n=t\frac{q-1}{p^e-1}+2$, $1\leq t\leq p^e-1$ & $k= \frac{n+p^e-1}{p^e+1}$ & \cite{RefJ (2023) Yang FFA} \\
        \hline
          $q=p^m$ odd, $2e\mid m$ & \tabincell{c}{$n=t\frac{q-1}{\gcd(x_2,q-1)}$, $1\leq t\leq \frac{q-1}{\gcd(x_1,q-1)}$,\\ $(q-1)\mid {\rm lcm}(x_1,x_2)$, $\frac{q-1}{p^e-1}\mid x_1$} & $1\leq k\leq \lfloor \frac{n+p^e-\frac{q-1}{\gcd(q-1,x_2)}}{p^e+1}\rfloor$ & \cite{RefJ (2023) Yang FFA} \\
        \hline
          $q=p^m$ odd, $2e\mid m$ & \tabincell{c}{$n=t\frac{q-1}{\gcd(x_2,q-1)}+1$, $1\leq t\leq \frac{q-1}{\gcd(x_1,q-1)}$,\\ $(q-1)\mid {\rm lcm}(x_1,x_2)$, $\frac{q-1}{p^e-1}\mid x_1$} &
          \tabincell{c}{$1\leq k\leq\max\{\lfloor \frac{n+p^e-1}{p^e+1}\rfloor,$\\$\lfloor \frac{n+p^{m-e}-1}{p^{m-e}+1}\rfloor\}$} & \cite{RefJ (2021) cao,RefJ (2023) cao}\\
        \hline
         $q=p^m$ odd, $2e\mid m$ & \tabincell{c}{$n=t\frac{q-1}{\gcd(x_2,q-1)}+2$, $1\leq t\leq \frac{q-1}{\gcd(x_1,q-1)}$,\\ $(q-1)\mid {\rm lcm}(x_1,x_2)$, $\frac{q-1}{p^e-1}\mid x_1$,\\ $(p^e+1)\mid (n-2)$} & $k= \frac{n+p^e-1}{p^e+1}$ & \cite{RefJ (2023) Yang FFA}\\
        \hline
         $q=p^m$ odd, $2e\mid m$ & \tabincell{c}{$n=t\tilde{m}$, $1\leq t\leq \frac{p^e-1}{m_1}$, \\$m_1=\frac{\tilde{m}}{\gcd(\tilde{m},y)}$, $y=\frac{q-1}{p^e-1}$, $\tilde{m}\mid (q-1)$} & $1\leq k\leq \lfloor \frac{n+p^e-\tilde{m}}{p^e+1}\rfloor$ & \cite{RefJ (2023) Yang FFA}\\
        \hline
         $q=p^m$ odd, $2e\mid m$ & \tabincell{c}{$n=t\tilde{m}+1$, $1\leq t\leq \frac{p^e-1}{m_1}$, \\$m_1=\frac{\tilde{m}}{\gcd(\tilde{m},y)}$, $y=\frac{q-1}{p^e-1}$, $\tilde{m}\mid (q-1)$} &
         \tabincell{c}{$1\leq k\leq\max\{\lfloor \frac{n+p^e-1}{p^e+1}\rfloor,$\\$\lfloor \frac{n+p^{m-e}-1}{p^{m-e}+1}\rfloor\}$} & \cite{RefJ (2021) cao,RefJ (2023) cao}\\
        \hline
         $q=p^m$ odd, $2e\mid m$ & \tabincell{c}{$n=t\tilde{m}+2$, $1\leq t\leq \frac{p^e-1}{m_1}$, $m_1=\frac{\tilde{m}}{\gcd(\tilde{m},y)}$,\\ $y=\frac{q-1}{p^e-1}$, $\tilde{m}\mid (q-1)$, $(p^e+1)\mid (n-2)$} & $ k= \frac{n+p^e-1}{p^e+1}$ & \cite{RefJ (2023) Yang FFA}\\
        \hline
          \tabincell{c}{$q=p^m$ odd, $2^h\mid \frac{m}{a}$,\\ $2^h=p^e+1$}&  \tabincell{c}{$n=tp^{az}$, $1\leq t\leq p^a$,\\ $1\leq z\leq \frac{m}{a}-1$} & $1\leq k\leq \lfloor \frac{n+p^e-1}{p^e+1}\rfloor$ & \cite{RefJ (2023) Yang CC}\\
        \hline
         $q=p^m$ odd, $2e\mid m$ &  \tabincell{c}{$n=tp^{az}$, $a\mid e$, $1\leq t\leq p^a$,\\ $1\leq z\leq \frac{m}{a}-1$} &
         \tabincell{c}{$1\leq k\leq\max\{\lfloor \frac{n+p^e-1}{p^e+1}\rfloor,$\\$\lfloor \frac{n+p^{m-e}-1}{p^{m-e}+1}\rfloor\}$} & \cite{RefJ (2021) cao,RefJ (2023) cao}\\
        \hline
        $q=p^m$ odd, $2e\mid m$ &  \tabincell{c}{$n=tp^{m-e}+1$, $1\leq t\leq p^e$,\\ $(p^e+1)\mid (n-2)$} & $ k= \frac{n+p^e-1}{p^e+1}$ & \cite{RefJ (2023) Yang FFA} \\
        \hline
         $q=p^m$ odd, $\frac{m}{e}$ odd &  \tabincell{c}{$GRS_{k'}(\bm{a},\bm{v})^{\perp_E}=GRS_{n-k'}(\bm{a},\bm{v})$,\\ $1\leq k'\leq \lfloor \frac{n}{2}\rfloor$} & $ 1\leq k\leq \lfloor \frac{n+p^e-1}{p^e+1}\rfloor$ & \cite{RefJ (2022) Fang CC} \\
        \hline
         \tabincell{c}{$q=p^m$, $\frac{m}{s}$ even,\\ $s=\gcd(e,m)$ } &  \tabincell{c}{$GRS_{k'}(\bm{a},\bm{v}) \subseteq GRS_{k'}(\bm{a},\bm{v})^{\perp_s}$ over $\F_q$, \\ $\deg(\lambda(x))\leq n-k'-1$} & $ 1\leq k\leq \lfloor \frac{n+p^e-1-\deg(\lambda(x))}{p^e+1}\rfloor$ & \cite{RefJ (2023) Yang CC} \\
        \hline
         \tabincell{c}{$q=p^m$ odd, $m$ even,\\ $\frac{m}{s}$ odd, $s=\gcd(e,m)$}  &  \tabincell{c}{$GRS_{k'}(\bm{a},\bm{v}) \subseteq GRS_{k'}(\bm{a},\bm{v})^{\perp_H}$ over $\F_q$, \\ $\deg(\lambda(x))\leq n-k'-1$} & $ 1\leq k\leq \lfloor \frac{n+p^e-1-\deg(\lambda(x))}{p^e+1}\rfloor$ & \cite{RefJ (2023) Yang CC} \\
        \hline
        $q=p^m$ even, $\frac{m}{s}$ odd & $n\leq q$, $m>1$ & $1\leq  k\leq \lfloor\frac{n+p^e-1}{p^e+1}\rfloor$ & \cite{RefJ (2024) Qian dcc} \\
        \hline
        $q=p^m>3$ & $n\leq p^a$, $a\mid m$, $(p^e+1)\mid\frac{q-1}{p^a-1}$ & $1\leq  k\leq \lfloor\frac{n+p^e-1}{p^e+1}\rfloor$ & \cite{RefJ (2024) Qian dcc} \\
        \hline
        $q=p^m>3$ & $n\mid q$ & $1\leq  k\leq \lfloor\frac{n+p^e-1}{p^e+1}\rfloor$ & \cite{RefJ (2024) Qian dcc} \\
        \hline
          $q=p^m>3$ & $(n-1)\mid (q-1)$, $n>1$ & $1\leq  k\leq \lfloor\frac{n+p^e-1}{p^e+1}\rfloor$ & \cite{RefJ (2024) Qian dcc} \\
        \hline
          $q=p^m>3$ & $n=2n'\leq q$, $n'\mid q$, $-1\in \langle w^{p^e+1}\rangle$& $1\leq  k\leq \lfloor\frac{n+p^e-1}{p^e+1}\rfloor$ & \cite{RefJ (2024) Qian dcc} \\
        \hline
         $q=p^m>3$ & \tabincell{c}{$n=tp^{az}$, $a\mid m$,  $1\leq t\leq p^a$,\\ $1\leq z\leq \frac{m}{a}-1$, $\gcd(p^e+1,q-1)\mid \frac{q-1}{p^a-1}$} & $1\leq k\leq  \lfloor\frac{n+p^e-1}{p^e+1}\rfloor$ & \cite{RefJ (2024) Qian dcc} \\
         \hline
         $q=p^m$ & $n=m$ & $1\leq k\leq \min\{e,m-e\}$ & \cite{RefJ (2023) Islam} \\
        \hline

	\end{tabular}}
\end{center}
\end{table}

\begin{table}
\caption{Our results on $e$-Galois self-orthogonal MDS codes.}
\label{tab:2}
\begin{center}
\resizebox{\textwidth}{68mm}{
	\begin{tabular}{c|ccc}
		\hline
		Finite field $q$ & Length $n$ & Dimension $k$ & References\\
		\hline
\tabincell{c}{$q=p^m$, \\ $\gcd(e,m)=\gcd(e',m)$}
        & \tabincell{c}{$GRS_{k'}(\bm{a},\bm{v})\subseteq GRS_{k'}(\bm{a},\bm{v})^{\perp_{e'}}$ over $\F_{q}$, \\for some  $p^{e'}\leq k'\leq \lfloor \frac{n+p^{e'}-1}{p^{e'}+1}\rfloor$,\\ $\deg(\lambda(x))\leq n-(k'-1)(p^{e'}+1)-2$} & $1\leq k\leq \lfloor \frac{n+p^e-1-\deg(\lambda(x))}{p^e+1}\rfloor$ & Lemma \ref{lem relation} \\
      \hline

 \multirow{4}{*}{\tabincell{c}{$q=p^m$ even, $\frac{m}{s}$ odd,\\ $s=\gcd(e,m)$}}
        & $1\leq n\leq \min \{p^s+1,q\}$  & $1\leq k\leq \lfloor \frac{n}{2}\rfloor$ & Corollary \ref{cor ps 2} (1)\\

        \cline{2-4}
        \multirow{4}{*}{}
       & \tabincell{c}{$p^s+1\leq n\leq p^{2s}$,\\ $n=\sum_{i=1}^tn_i$, $n_i\leq p^s$, $t\leq p^s$} & \tabincell{c}{$1\leq k\leq \max\{\lfloor \frac{n+p^e-1}{p^e+1}\rfloor,$\\$ \min\{\lfloor \frac{n_1}{2}\rfloor,\lfloor \frac{n_2}{2}\rfloor,\dots,\lfloor \frac{n_t}{2}\rfloor \} \}$ }& Theorem \ref{th p2s 1} (1) \\

       \cline{2-4}

       \multirow{4}{*}{}
       & \tabincell{c}{$n=r(p^e-1)+1$, \\$s=e$, $1\leq r\leq p^e+1$}  & $1\leq k\leq \lfloor \frac{p^e+r-1}{2}\rfloor$ & Theorem \ref{th 11 3333}\\

       \cline{2-4}

       \multirow{4}{*}{}
       & $n=q+1$ & $1\leq k\leq \lfloor \frac{q+p^e-2}{p^e+1}\rfloor$ & Theorem \ref{th 2 q+1} \\
       \hline

    \multirow{9}{*}{\tabincell{c}{$q=p^m$ odd, $\frac{m}{s}$ odd,\\ $s=\gcd(e,m)$}}
        & \tabincell{c}{$1\leq n\leq p^s$, \\$GRS_{k'}(\bm{a},\bm{v})\subseteq GRS_{k'}(\bm{a},\bm{v})^{\perp_E}$ over $\F_{p^s}$} & $1\leq k\leq k'$ & Corollary \ref{cor ps 2} (2)\\

        \cline{2-4}
        \multirow{9}{*}{}
       & \tabincell{c}{$2\leq n\leq p^s+1$, \\$GRS_{k'}(\bm{a},\bm{v},\infty)\subseteq GRS_{k'}(\bm{a},\bm{v},\infty)^{\perp_E}$ over $\F_{p^s}$ }& $1\leq k\leq k'$ & Corollary \ref{cor ps 2} (2) \\

       \cline{2-4}

       \multirow{9}{*}{}
       & \tabincell{c}{$p^s+1\leq n\leq p^{2s}$, \\$n=\sum_{i=1}^tn_i$, $n_i\leq p^s$, $t\leq p^s$,\\ $GRS_{k_i}(\bm{a}_i,\bm{v}_i)\subseteq GRS_{k_i}(\bm{a}_i,\bm{v}_i)^{\perp_E}$ over $\F_{p^s}$}   & $1\leq k\leq \min\{k_1,k_2,\dots,k_t\} $ & Theorem \ref{th p2s 1} (2) \\

       \cline{2-4}

       \multirow{9}{*}{}
       & \tabincell{c}{$1\leq n\leq q$, \\$GRS_{k'}(\bm{a},\bm{v})\subseteq GRS_{k'}(\bm{a},\bm{v})^{\perp_E}$ over $\F_{q}$} & $1\leq k\leq 1+\lfloor \frac{2k'-2}{p^e+1}\rfloor$ & Corollary \ref{cor relation}\\

        \cline{2-4}
       \multirow{9}{*}{}
       & $n=q+1$ & $1\leq k\leq \lfloor \frac{q+p^e-4}{p^e+1}\rfloor$ & Theorem \ref{th 2 q+1} (2)\\

       \hline

     \multirow{5}{*}{\tabincell{c}{$q=p^m$, $\frac{m}{s}$ even,\\ $s=\gcd(e,m)$}}
        & $1\leq n\leq p^s+1$ & $1\leq k\leq \lfloor \frac{n}{2}\rfloor$  & Corollary \ref{cor ps 2} (3)\\

        \cline{2-4}
        \multirow{5}{*}{}
       &  \tabincell{c}{ $1 \leq n\leq p^{2s}$, \\$GRS_{k'}(\bm{a},\bm{v})\subseteq GRS_{k'}(\bm{a},\bm{v})^{\perp_H}$ over $\F_{p^{2s}}$ }& $1\leq k\leq k'$ & Theorem \ref{th p2s 1} (3)\\

       \cline{2-4}

       \multirow{5}{*}{}
       & \tabincell{c}{ $2 \leq n\leq p^{2s}+1$, \\$GRS_{k'}(\bm{a},\bm{v},\infty)\subseteq GRS_{k'}(\bm{a},\bm{v},\infty)^{\perp_H}$ over $\F_{p^{2s}}$} & $1\leq k\leq k'$ & Theorem  \ref{th p2s 1} (3)\\

        \cline{2-4}
       \multirow{5}{*}{}
       & $n=q+1$ & \tabincell{c}{$1\leq k\leq \lfloor \frac{q+p^e-2(p^s+1)}{p^e+1}\rfloor$ or \\ $k=\frac{q-1}{p^e+1}+1$ $(s=e)$} & Theorem \ref{th 2 q+1} (3)\\

        \hline

	\end{tabular}}
\end{center}
\end{table}

\subsection{Our results}

In this paper, we focus on constructing new Galois self-orthogonal MDS codes, in particular, $e$-Galois self-orthogonal MDS codes with dimension greater than $\lfloor \frac{n+p^e-1}{p^e+1}\rfloor$.
Moreover, we give some direct and effective methods for the construction of Galois self-orthogonal codes.
The main contributions of this paper can be summarized as follows:

\begin{itemize}
\item
For a vector $\bm{v}\in (\F_q^*)^n$ and a $q$-ary linear code $\C$ with generator matrix $G$,
we give some necessary and sufficient conditions for the equivalent code $\Phi_{\bm{v}}(\C)$ of $\C$ and the extended code $\widetilde{\Phi_{\bm{v}}(\C)}(G diag(\bm{v}),\infty)$ (see Definition \ref{def Extend}) of $\Phi_{\bm{v}}(\C)$ to
be $e$-Galois self-orthogonal (see Lemmas \ref{lem galois oa} and \ref{lem galois oa extend}).
This provides a general framework for constructing Galois self-orthogonal codes.
From this, we directly obtain some necessary and sufficient conditions for GRS and EGRS codes to be Galois self-orthogonal (see Lemmas \ref{lem <a,v>} and \ref{lem <a,v> EG}).

\item
Necessary and sufficient conditions for GRS and EGRS code to
be Galois self-orthogonal are presented when $p^e\leq k\leq \lfloor \frac{n+p^e-1}{p^e+1}\rfloor$, where $2e\leq m$ (see Lemmas \ref{lem lam(x)} and \ref{lem lam(x) 2}).
This is a generalisation of the results in \cite{RefJ (2021) Fang It} and an improvement of the results in \cite{RefJ (2023) Yang FFA}.
From this, we give a class of $e$-Galois self-orthogonal MDS codes obtained from $e'$-Galois self-orthogonal GRS codes,
where $\gcd(e,m)=\gcd(e',m)$ (see Lemma \ref{lem relation}).

\item
In order to construct $e$-Galois self-orthogonal MDS codes with dimension greater than  $\lfloor \frac{n+p^e-1}{p^e+1}\rfloor$, we consider $e$-Galois self-orthogonal (extended) GRS codes of length $n\leq p^{2e}$.

\begin{itemize}
\item
\textbf{Case of length} $1\leq n\leq p^s+1$: in this case, we construct some new $e$-Galois self-orthogonal MDS codes
that can have dimensions from $1$ to $\lfloor \frac{n}{2}\rfloor$
(see Theorem \ref{th 1 ax+b} and Corollary \ref{cor ps 2}).
Therefore, when $\frac{m}{s}$ is odd and $p$ is even, or $\frac{m}{s}$ is even, we completely determine the $q$-ary $e$-Galois self-orthogonal MDS codes of length $1\leq n\leq p^s+1$ for all possible parameters;

\item
\textbf{Case of length} $p^s+1 \leq n\leq p^{2s}+1$:
in this case, we construct some new $e$-Galois self-orthogonal MDS codes
that can have dimensions greater than $\lfloor \frac{n+p^e-1}{p^e+1}\rfloor$
(see Theorems \ref{th p2s 1} and \ref{th 11 3333}).
In particular, when $s=e$, the dimension of some $e$-Galois self-orthogonal (extended) GRS codes can take the value $p^e-1$.

\end{itemize}

\item
Moreover, we construct some new $q$-ary $e$-Galois self-orthogonal MDS codes of \textbf{length} $q+1$ (see Theorem \ref{th 2 q+1}).
The dimension of our  $e$-Galois self-orthogonal EGRS code is close to or equal to the value $\lceil\frac{q+p^e+1}{p^e+1}\rceil-1$.

\item
We list the new Galois self-orthogonal MDS codes we constructed in Table \ref{tab:2} of Section \ref{sec5}.
By comparison, it can be shown that the Galois self-orthogonal MDS codes obtained in this paper have new parameters.
Moreover, by propagation rules, we obtain some new MDS codes
with Galois hulls of arbitrary dimensions.
Then many quantum codes can be obtained from these MDS codes with Galois
hulls.

\end{itemize}

\subsection{Organization of this paper}

The rest of this paper is organized as follows.
In Section \ref{sec2}, we recall some basic notations about Galois self-orthogonal (extended) GRS codes.
In Section \ref{sec3},
we give the main methods of this paper.
In Section \ref{sec4}, we construct many new classes of Galois self-orthogonal MDS codes.
In Section \ref{sec5}, we compare the results of this paper with previous results and use propagation rules to obtain new MDS codes with Galois hulls.
Finally, Section \ref{sec6} concludes this paper.

\section{Preliminaries}\label{sec2}

Throughout this paper, we fix some notations as follows for convenience.

\begin{itemize}
\item
Let $q=p^m$, where $p$ is prime and $\F_q$ be the finite field with $q$ elements.
Set $\F_q^*=\F_q\backslash \{0\}$.
\item
$k$ and $n$ are both positive integers with $1\leq k\leq n$.
\item
$e$ is an integer with $0\leq e\leq m-1$. Let $s=\gcd(e,m)=\gcd(m-e,m)$.

\item
For any finite set $S$, we denote $|S|$ as the number of all elements in $S$.

\item
Fix $w$ as a primitive element of $\F_q$, i.e. $\F_q^*=\langle w\rangle$.
The multiplicative order of $a$ in $\F_q^*$ is denoted by $ord(a)$.
Let $H=\langle w^{p^e+1}\rangle$ be the cyclic group generated by $w^{p^e+1}$,
it follows that $|H|=ord(w^{p^e+1})=\frac{q-1}{\gcd(q-1,p^e+1)}$.
Let $QR$ be the set of nonzero squares of $\F_q$.

\item
Let $\bm{0}=(0,0,\dots,0)$, $\bm{1}=(1,1,\dots,1)$, $\infty=(0,\dots,0,1)$
and the lengths of $\bm{0}$, $\bm{1}$ and $\infty$ depend on the
context.
\item

For an $k\times n$ matrix $G$ over $\F_q$ with row vectors $\bm{g}_1,\bm{g}_2,\dots,\bm{g}_k\in \F_q^n$, write $G=(\bm{g}_1,\bm{g}_2,\dots,\bm{g}_k)^{\mathcal{T}}$.

\item
For any vector $\bm{c}=(c_1,c_2,\dots,c_n)\in \F_{q}^n$ and integer $i$, denote
$\bm{c}^i=\begin{cases}
(1,1,\dots,1), & {\rm if}\ i=0;\\
(c_1^i,c_2^i,\dots,c_n^i), & {\rm if}\ i\neq 0.\\
\end{cases}$

\item
For any vector $\bm{a}=(a_{1},a_{2},\dots,a_{n})\in \F_q^n$
with $a_1,a_2,\dots,a_n$ distinct elements, denote
$$ u_i=\prod_{1\leq j\leq n,j\neq i}(a_i-a_j)^{-1}\quad {\rm and}\quad \bm{u}=(u_1,u_2,\dots,u_n)\in (\F_q^*)^n \ {\rm for}\ 1\leq i\leq n.$$

\item
Let $\sigma$: $\F_q\rightarrow\F_q$, $a\mapsto a^p$ be the \emph{Frobenius automorphism} of $\F_q$.
For any vector $\bm{x}=(x_1,x_2,\dots,x_n)\in \F_q^n$ and any matrix $G=(g_{ij})_{k\times n}$ over $\F_q$,
we denote
$$\sigma(\bm{x})=(\sigma(x_1),\sigma(x_2),\dots,\sigma(x_n))\quad {\rm and}\quad \sigma(G)=(\sigma(g_{ij}))_{k\times n}.$$
It follows that the mapping $\sigma^e$: $\F_q\rightarrow\F_q$, $\sigma^e(a)=a^{p^e}$, $\forall\ a\in \F_q$
is an automorphism of $\F_q$.
Moreover, the inverse of the mapping
$\sigma^e$ is denoted by $\sigma^{m-e}$: $\sigma^{m-e}(a)=a^{p^{m-e}}$.

\item

For any two vectors $\bm{x}=(x_1,x_2,\dots,x_n)\in \F_{q}^n$ and $\bm{y}=(y_1,y_2,\dots,y_n)\in \F_{q}^n$,
the \emph{Schur product} between $\bm{x}$ and $\bm{y}$ is defined as
$$\bm{x}\star \bm{y}=(x_1y_1,x_2y_2,\dots,x_ny_n)\in \F_q^n.$$
Moreover, the \emph{Schur product} of two codes $\C$ and $\C'$ with length $n$ is defined as
$$\C\star \C'=\big \langle \bm{c}\star \bm{c}': \bm{c}\in \C,\bm{c}'\in \C' \big \rangle.$$

\item
For a $q$-ary linear code $\C$ with length $n$ and vector $\bm{v}\in (\F_q^*)^n$,
denote
$$\Phi_{\bm{v}}(\C)= \big\{\bm{v}\star \bm{c}: \bm{c}\in \C \big\}.$$
It is easy to check that $\C$ and $\Phi_{\bm{v}}(\C)$ are \emph{linearly equivalent}.
\end{itemize}

\subsection{Galois dual codes and Galois hulls}

For any two vectors $\bm{x}=(x_1,x_2,\dots,x_n),\ \bm{y}=(y_1,y_2,\dots,y_n)\in \F_{q}^n$,
the \emph{$e$-Galois inner product} of vectors $\bm{x}$ and $\bm{y}$ is defined as
$$\langle \bm{x},\bm{y} \rangle_{e}=\sum_{i=1}^{n}x_iy_i^{p^e},$$
where $e$ is an integer with $0\leq e\leq m-1$.
In a natural view, we can regard the $m$-Galois inner product as the $0$-Galois inner product.
The $e$-Galois inner product is a generalization of the Euclidean inner product ($e=0$) and the Hermitian inner product ($e=\frac{m}{2}$ with even $m$).
For convenience, we use $\langle - ,- \rangle_E$ and $\langle -,- \rangle_H$ to denote $\langle -, -\rangle_{0}$ and $\langle -,-\rangle_{\frac{m}{2}}$ ($m$ even), respectively.
Let $\C$ be a linear code over $\F_q$, then the \emph{$e$-Galois dual code} of $\C$ is defined as
$$\C^{\perp_{e}}=\Big\{\bm{x}\in \F_{q}^n:\langle \bm{y}, \bm{x}\rangle_{e}=0, \ {\rm for \ all}\ \bm{y}\in \C \Big \}.$$
It follows that $\C^{\bot_0}$ and $\C^{\bot_{\frac{m}{2}}}$ ($m$ even) are just the Euclidean and Hermitian dual codes of $\C$, respectively.
For convenience, we denote them as $\C^{\bot_E}$ and $\C^{\bot_H}$, respectively.
Denote $\Hull_e(\C)=\C\cap \C^{\bot_e}$ as the \emph{$e$-Galois hull} of $\C$.
In particular, if $\C\subseteq \C^{\bot_e}$ (i.e., $\Hull_e(\C)=\C$),
we say that $\C$ is \emph{$e$-Galois self-orthogonal.}
It is worth noting that Galois self-orthogonal codes have the following properties.

\begin{lemma}\label{lem duichen}
(\cite[Lemma 2.4]{RefJ (2023) Yang DM})
Let $\C$ be an $[n,k,d]_q$ code. Then $\C$ is an $e$-Galois self-orthogonal code if and only if $\C$ is an $(m-e)$-Galois self-orthogonal code.
\end{lemma}

\begin{remark}\label{rem 1}
By Lemma \ref{lem duichen}, we have the following results.
\begin{itemize}
\item

For the construction of $e$-Galois self-orthogonal codes, we only discuss the case $2e\leq m$ without loss of generality.
\item
If there exists an $[n_{e},k_{e}]_q$ $e$-Galois self-orthogonal MDS code, where $n_{e}$, $k_{e}$ are related to $e$,
then there exists an $[n_{e},k_{e}]_q$ $(m-e)$-Galois self-orthogonal MDS code.
\end{itemize}
\end{remark}

\subsection{GRS and EGRS codes}

\begin{definition}\label{def GRS}
Let $\bm{a}=(a_{1},a_{2},\dots,a_{n})\in \F_q^n$
with $a_1,a_2,\dots,a_n$ distinct elements and $\bm{v}=(v_1,v_2,\dots,v_n)\in (\F_{q}^*)^n$.
For an integer $1\leq k\leq n$, the GRS code associated with $\bm{a}$ and $\bm{v}$ is defined as
$$GRS_{k}(\bm{a},\bm{v})=\Big\{(v_{1}f(a_{1}),v_{2}f(a_{2}),\dots,v_{n}f(a_{n})): f(x)\in \F_{q}[x],\ \deg(f(x))\leq k-1 \Big\}.$$
The code $GRS_k(\bm{a},\bm{v})$ has a generator matrix
$$ G_{k}(\bm{a},\bm{v})=(\bm{a}^0\star \bm{v},\bm{a}^1\star \bm{v},\dots,\bm{a}^{k-1}\star \bm{v})^\mathcal{T}. $$

Moreover, the $k$-dimensional EGRS code of length $n+1$ associated to $\bm{a}$ and $\bm{v}$ is defined as
$$GRS_{k}(\bm{a},\bm{v},\infty)=\Big \{(v_{1}f(a_{1}),v_{2}f(a_{2}),\dots,v_{n}f(a_{n}),f_{k-1}):\ f(x)\in \F_{q}[x],\ \deg(f(x))\leq k-1 \Big\},$$
where $f_{k-1}$ is the coefficient of $x^{k-1}$ in $f(x)$.
The code $GRS_k(\bm{a},\bm{v},\infty)$ has a generator matrix
$$G_{k}(\bm{a},\bm{v},\infty)=[(\bm{a}^0\star \bm{v},\bm{a}^1\star \bm{v},\dots,\bm{a}^{k-1}\star \bm{v})^{\mathcal{T}}|\infty^T],$$
where $\infty=(0,0,\dots,1)\in \F_{q}^k$.
\end{definition}

The elements $a_1, a_2,\dots,a_n$ are called the \emph{code locators} of
$GRS_k(\bm{a},\bm{v})$ (resp. $GRS_k(\bm{a},\bm{v},\infty)$).
It is well known that GRS codes, EGRS codes and their dual codes are MDS codes.
There are some important results in the literature on GRS codes and EGRS codes. We review them here.

\begin{lemma}\label{lem equiv}
(\cite[Corollary 3.3]{RefJ (2023) Liu Galois})
Let the two vectors $\bm{a}$ and $\bm{v}$ be defined as in Definition \ref{def GRS}.
For any $a\in \F_q^*$, $b\in \F_q$ and $\lambda\in \F_q^*$,
we have $$GRS_k(\bm{a},\bm{v})=GRS_k(a\bm{a}+b\bm{1},\lambda \bm{v})$$
and
$$GRS_k(\bm{a},\bm{v},\infty)=GRS_k(a\bm{a}+b\bm{1},a^{1-k}\bm{v},\infty),$$
where $a\bm{a}+b\bm{1}=(aa_1+b,aa_2+b,\dots,aa_n+b)\in \F_q^n$ and $\lambda\bm{v}=(\lambda v_1, \lambda v_2, \dots, \lambda v_n)\in (\F_q^*)^n$.
\end{lemma}

\begin{remark}\label{rem 2}
By Lemma \ref{lem equiv}, we have the following results.

\begin{itemize}
\item
By taking $a=1$, $b=-a_n$ and $\lambda=v_n^{-1}$
in Lemma \ref{lem equiv}, then $aa_n-b=0$ and $\lambda v_n=1$, thus for any GRS code $GRS_k(\bm{a},\bm{v})$, we can always assume $a_n=0$ and $v_n=1$
without loss of generality.
\item
For an EGRS code $GRS_k(\bm{a},\bm{v},\infty)$ of length $n+1\leq q$, there exists $b\in \F_q$ such that $\bm{a}+b\bm{1}\in (\F_q^*)^{n}$, thus we can always assume that $\bm{a}\in (\F_q^*)^{n}$ without loss of generality.
\end{itemize}
\end{remark}

\begin{lemma}\label{lem GRS=EGRS}
(\cite{RefJ (2022) equivalence})
Let $\bm{a}=(a_{1},a_{2},\dots,a_{n})\in (\F_q^*)^n$
with $a_1,a_2,\dots,a_n$ distinct elements and $\bm{v}=(v_1,v_2,\dots,v_n)\in (\F_{q}^*)^n$.
Then we have
$$GRS_k(\bm{a},\bm{v},\infty)=GRS_k((\bm{a}^{-1},0),(\bm{v}\star \bm{a}^{k-1} ,1)).$$
\end{lemma}

From Lemma \ref{lem GRS=EGRS}, we have the following corollary which gives the equivalence of GRS codes and EGRS codes when $n\leq q$.

\begin{corollary}\label{cor GRS=EGRS}
(\cite{RefJ (1996) Pellikaan})
If $n\leq q$, then a linear code with length $n$ is GRS if and only if it is EGRS.
\end{corollary}

It is easy to verify the following lemma, which gives some subcodes of GRS codes and these subcodes are still GRS codes.

\begin{lemma}\label{lem kkkk}
For $1\leq k \leq k'\leq n$, the following statements hold.
\begin{itemize}
\item[(1)]  $GRS_k(\bm{a},\bm{v})\subseteq GRS_{k'}(\bm{a},\bm{v})$.
\item[(2)] If $\bm{a}\in (\F_q^*)^n$, $GRS_k(\bm{a},\bm{a}^{i}\star \bm{v})\subseteq GRS_{k'}(\bm{a},\bm{v})$, where $0\leq i\leq k'-k$.
\item[(3)] If $\bm{a}\in (\F_q^*)^n$,  $GRS_k(\bm{a}, \bm{a}^{k'-k}\star \bm{v}, \infty)\subseteq GRS_{k'}(\bm{a},\bm{v},\infty).$
\end{itemize}
\end{lemma}

Note that for $1\leq k \leq k'\leq n$, if $GRS_{k'}(\bm{a},\bm{v})$ is $e$-Galois self-orthogonal, then $GRS_k(\bm{a},\bm{v})$ is also $e$-Galois self-orthogonal. By Lemma \ref{lem kkkk}, we have the following corollary.

\begin{corollary}\label{cor kkkk}
If there exists an $[n,k']_{q}$
$e$-Galois self-orthogonal (extended) GRS code, where $n\leq q$,
then for $1\leq k\leq k'$, there exists an $[n,k]_q$ $e$-Galois self-orthogonal (extended) GRS code.
\end{corollary}

\begin{remark}
Corollary \ref{cor kkkk} states that for the construction of an $e$-Galois self-orthogonal (extended) GRS code of length $n\leq q$, we only need to determine the maximum number of dimensions $k$ that can be attained with length $n$.
\end{remark}

\subsection{Extended code}

Inspired by the GRS and EGRS codes, we give the following code extending method. This is a common extending method.

\begin{definition}\label{def Extend}
For a given $[n,k]_q$ linear code $\C$ with generator matrix $G$,
construct a $k\times(n+1)$ matrix
 $$\widetilde{G}(G,\infty)=[G|\infty^T],$$
 where $\infty=(0,\dots,0,1)\in \F_q^k$.
Let $\widetilde{\C}(G,\infty)$ be the linear code with generator matrix $\widetilde{G}(G,\infty)$.
By definition, the extended code $\widetilde{\C}(G,\infty)$ has parameters $[n,k,\tilde{d}]_q$ where $\tilde{d}=d$ or $\tilde{d}=d+1$.
\end{definition}

\begin{remark}
By the definition of extended code, it is easy to notice the following:
\begin{itemize}
\item Two extended codes $\widetilde{\C}(G_1,\infty)$ and $\widetilde{\C}(G_2,\infty)$ may not be equivalent, where $G_1$ and $G_2$ are two generator matrices of $\C$.
\item $GRS_k(\bm{a},\bm{v},\infty)=\widetilde{GRS_k(\bm{a},\bm{v})}(G_k(\bm{a},\bm{v}),\infty)$.
\item Let $\bm{v}=(v_1,v_2,\dots,v_n)\in (\F_q^*)^n$ and $(\bm{v},v_{n+1})\in (\F_q^*)^{n+1}$.
For a given $[n,k]_q$ linear code $\C$ with generator matrix $G$, we have
 $$\Phi_{(\bm{v},v_{n+1})}(\widetilde{\C}(G,\infty))=\widetilde{\Phi_{v_{n+1}^{-1}\bm{v}}(\C)}(v_{n+1}^{-1} G diag(\bm{v}),\infty).$$
\end{itemize}
\end{remark}

\section{Main methods}\label{sec3}

In this section, we give many methods to construct Galois self-orthogonal codes.
Our approach provides a general framework that effectively unifies similar known techniques for constructing Galois self-orthogonal codes.
We first give the following lemmas, which are useful in the proof of the main results.

\begin{lemma}\label{lem gcd()}
(\cite[Lemma 3]{RefJ (2017) gcd Li})
Let $m\geq 1$ and $p>1$ be two integers. Then
\[\gcd(p^e+1,p^m-1)=\begin{cases}
1, & if\ \frac{m}{s}\ is\ odd\ and\ p\ is\ even;\\
2, & if\ \frac{m}{s}\ and\ p\ are\ odd;\\
p^{s}+1, & if\ \frac{m}{s}\ is\ even.\\
\end{cases}\]
\end{lemma}

By Lemma \ref{lem gcd()}, we give the following lemma, which gives the specific $H$.

\begin{lemma}\label{lem HHH}
Let $m\geq 1$ and $p>1$ be two integers. Then
\[ H=\begin{cases}
\F_q^*, & if\ \frac{m}{s}\ is\  odd\ and\ p\ is\  even;\\
QR, & if\ \frac{m}{s}\ and\ p\ are\ odd;\\
\langle w^{p^{s}+1}\rangle , & if\ \frac{m}{s}\ is\ even.\\
\end{cases}\]
\end{lemma}

\begin{proof}
By Lemma \ref{lem gcd()},
we have that
$$|H|=ord(w^{p^e+1})=\frac{q-1}{\gcd(q-1,p^e+1)}=\begin{cases}
q-1=|\F_q^*|, & {\rm if}\ \frac{m}{s}\ {\rm is\  odd\ and}\ p\ {\rm is\  even};\\
\frac{q-1}{2}=|QR|, & {\rm if}\ \frac{m}{s}\ {\rm and}\ p\ {\rm are\ odd};\\
\frac{q-1}{p^s+1}=|\langle w^{p^s+1}\rangle|, & {\rm if}\ \frac{m}{s}\ {\rm is\ even}.\\
\end{cases}$$
Note that
$$H= \langle w^{p^e+1}\rangle \subseteq \begin{cases}
\langle w\rangle= \F_{q}^*, & {\rm if}\ \frac{m}{s}\ {\rm is\  odd\ and}\ p\ {\rm is\  even};\\
\langle w^2\rangle= QR, & {\rm if}\ \frac{m}{s}\ {\rm and}\ p\ {\rm are\ odd};\\
\langle w^{p^s+1}\rangle, & {\rm if}\ \frac{m}{s}\ {\rm is\ even},\\
\end{cases}$$
it is not difficult to deduce the desired result.
This completes the proof.
\end{proof}

\subsection{Necessary and sufficient
conditions for the code $\Phi_{\bm{v}}(\C)$ to be Galois self-orthogonal}

In this subsection, we give necessary and sufficient conditions for the code $\Phi_{\bm{v}}(\C)$ to be Galois self-orthogonal.
Then we directly obtain necessary and sufficient conditions for GRS codes to be Galois self-orthogonal.

\begin{lemma}\label{lem galois oa}
Let $\C$ be an $[n,k]_{q}$ linear code with its dimension $k\leq \frac{n}{2}$.
Let $\C=\langle \bm{c}_i: 1\leq i\leq k\rangle$ and $\bm{v}\in (\F_{q}^*)^n$, where $\bm{c}_i\in \F_{q}^n\ (1\leq i\leq k)$.
Suppose that
$$C=(\bm{c}_1\star \sigma^e(\bm{c}_1),\dots,\bm{c}_1\star \sigma^e(\bm{c}_k),
\bm{c}_2\star \sigma^e(\bm{c}_1),\dots,\bm{c}_2\star \sigma^e(\bm{c}_k), \dots,
\bm{c}_k\star \sigma^e(\bm{c}_1),\dots,\bm{c}_k\star \sigma^e(\bm{c}_k))^{\mathcal{T}}$$
is a $k^2\times n$ matrix over $\F_{q}$.
Then the following statements are equivalent.
\begin{itemize}
\item[(1)] $\Phi_{\bm{v}}(\C)$ is $e$-Galois self-orthogonal.
\item[(2)] $\Phi_{\bm{v}}(\C)$ is $(m-e)$-Galois self-orthogonal.
\item[(3)] $\bm{v}^{p^e+1}\in (\C\star \sigma^e(\C))^{\bot_E}$.
\item[(4)] $\bm{v}^{p^{m-e}+1}\in (\C\star \sigma^{m-e}(\C))^{\bot_E}$.
\item[(5)] The  equation $C\bm{x}^T=\bm{0}^T$ has a solution $\bm{x}=(x_1,x_2,\dots,x_{n})\in H^{n}$, where $\bm{x}=\bm{v}^{p^e+1}$.
\item[(6)] The equation $\sigma^{m-e}(C)\bm{y}^T=\bm{0}^T$
has a solution $\bm{y}=(y_1,y_2,\dots,y_{n})\in H^{n}$, where $\bm{y}=\sigma^{m-e}(\bm{x})$.
\end{itemize}
\end{lemma}

\begin{proof}

By the definition,
$\Phi_{\bm{v}}(\C)=\langle \bm{c}_i\star \bm{v}: 1\leq i\leq k\rangle$.
Then $\Phi_{\bm{v}}(\C)$ is $e$-Galois self-orthogonal if and only if
$$\langle \bm{c}_i\star \bm{v},\bm{c}_j\star \bm{v}  \rangle_e=0,\ {\rm for\ any}\ 1\leq i,j\leq k,$$
if and only if
$$\langle \bm{c}_i\star \bm{v}, \sigma^e(\bm{c}_j\star \bm{v})  \rangle_E=\langle \bm{c}_i\star \sigma^e(\bm{c}_j),\bm{v}^{p^e+1}  \rangle_E=0,\ {\rm for\ any}\ 1\leq i,j\leq k.$$
Note that
$$\C\star \sigma^e(\C)= \langle \bm{c}\star \bm{c}': \bm{c}\in \C,\bm{c}'\in \sigma^e(\C) \rangle=
\langle \bm{c}_i\star \sigma^e(\bm{c}_j): 1\leq i,j\leq k\rangle,$$ it follows that
 $\Phi_{\bm{v}}(\C)$ is $e$-Galois self-orthogonal
if and only if $\bm{v}^{p^e+1}\in (\C \star \sigma^e(\C))^{\bot_E}$.
Since $H=\langle w^{p^e+1}\rangle$,
then for any $x\in \F_q^*$, $x\in H$ if and only if there exists $v\in \F_q^*$ such that $x=v^{p^e+1}$.
Hence, we have that $\Phi_{\bm{v}}(\C)$ is $e$-Galois self-orthogonal
if and only if $C\bm{x}^T=\bm{0}^T$, where $\bm{x}=\bm{v}^{p^e+1}$.
The desired result follows from Lemma \ref{lem duichen}.
This completes the proof.
\end{proof}

Note that $GRS_k(\bm{a},\bm{v})=\Phi_{\bm{v}}(GRS_k(\bm{a},\bm{1}))$ and
$GRS_k(\bm{a},\bm{1})=\langle \bm{a}^0,\bm{a}^1,\dots,\bm{a}^{k-1}\rangle$.
By Lemma \ref{lem galois oa}, we can directly obtain the following lemma, which gives necessary and sufficient conditions for GRS codes to be Galois self-orthogonal.

\begin{lemma}\label{lem <a,v>}
Let the two vectors $\bm{a}$ and $\bm{v}$ be defined as in Definition \ref{def GRS}.
Suppose that
$$A=(\bm{a}^0,\dots,\bm{a}^{k-1},\bm{a}^{p^e},\dots,\bm{a}^{p^e+k-1},\dots,\bm{a}^{(k-1)p^e},\dots,\bm{a}^{(k-1)p^e+k-1})^\mathcal{T}$$
is a $k^2\times n$ matrix over $\F_{q}$.
Then the following statements are equivalent.
\begin{itemize}
\item[(1)] $GRS_k(\bm{a},\bm{v})$ is $e$-Galois self-orthogonal.
\item[(2)] $GRS_k(\bm{a},\bm{v})$ is $(m-e)$-Galois self-orthogonal.
\item[(3)] $\langle \bm{a}^{p^ei+j}, \bm{v}^{p^e+1} \rangle_E=0$, for $0\leq i,j \leq k-1$.
\item[(4)] $\langle \bm{a}^{p^{m-e}i+j}, \bm{v}^{p^{m-e}+1} \rangle_E=0$, for $0\leq i,j \leq k-1$.
\item[(5)] The equation $A\bm{x}^T=\bm{0}^T$ has a solution $\bm{x}=(x_1,x_2,\dots,x_{n})\in H^{n}$, where $\bm{x}=\bm{v}^{p^e+1}.$
\item[(6)] The equation $\sigma^{m-e}(A)\bm{y}^T=\bm{0}^T$
has a solution $\bm{y}=(y_1,y_2,\dots,y_{n})\in H^{n}$, where $\bm{y}=\sigma^{m-e}(\bm{x})$.
\end{itemize}
\end{lemma}

\begin{remark}
If we take $e=\frac{m}{2}$ ($m$ even) in Lemma \ref{lem <a,v>} (3), we obtain the results presented in
\cite[Lemma 3.1]{RefJ (2014) L.jin}.
In other words, Lemma \ref{lem <a,v>} is a generalization of \cite[ Lemma 3.1]{RefJ (2014) L.jin}.
\end{remark}

\subsection{Necessary and sufficient
conditions for the code $\widetilde{\Phi_{\bm{v}}(\C)}(G diag(\bm{v}),\infty)$ to be Galois self-orthogonal}

In this subsection, we give necessary and sufficient conditions for the code $\widetilde{\Phi_{\bm{v}}(\C)}(G diag(\bm{v}),\infty)$ to be Galois self-orthogonal.
Then we obtain necessary and sufficient conditions for EGRS codes to be Galois self-orthogonal.

\begin{lemma}\label{lem galois oa extend}
Let $\C$ be an $[n,k]_q$ linear code with
generator matrix $G=(\bm{c}_1,\bm{c}_2,\dots,\bm{c}_k)^{\mathcal{T}}$ and $\bm{v}\in (\F_{q}^*)^n$,
where $\bm{c}_i\in \F_{q}^n\ (1\leq i\leq k)$.
Suppose that the matrix $C$ is defined as in Lemma \ref{lem galois oa} and $\bm{b}=(0,\dots,0,-1)\in \F_q^{k^2}$.
Then the following statements are equivalent.
\begin{itemize}
\item[(1)] $\widetilde{\Phi_{\bm{v}}(\C)}(G diag(\bm{v}),\infty)$ is $e$-Galois self-orthogonal.
\item[(2)] $\widetilde{\Phi_{\bm{v}}(\C)}(G diag(\bm{v}),\infty)$ is $(m-e)$-Galois self-orthogonal.
\item[(3)] $\left\{
             \begin{array}{lr}
            \langle \bm{c}_i\star \sigma^e(\bm{c}_j), \bm{v}^{p^e+1}  \rangle_E=0,\ for\  0\leq i,\ j\leq k-1,\  except\  i=j=k-1;&  \\
             \langle \bm{c}_{k+1}^{p^e+1},\bm{v}^{p^e+1}  \rangle_E=-1.&
             \end{array}
\right.$
\item[(4)] $\left\{
             \begin{array}{lr}
            \langle \bm{c}_i\star \sigma^{m-e}(\bm{c}_j), \bm{v}^{p^{m-e}+1}  \rangle_E=0,\ for\ 0\leq i,\ j\leq k-1,\  except\  i=j=k-1;&  \\
             \langle \bm{c}_{k+1}^{p^{m-e}+1},\bm{v}^{p^{m-e}+1}  \rangle_E=-1.&
             \end{array}
\right.$
\item[(5)] The equation $C\bm{x}^T=\bm{b}^T$ has a solution $\bm{x}=(x_1,x_2,\dots,x_{n})\in H^{n}$, where $\bm{x}=\bm{v}^{p^e+1}.$
\item[(6)] The equation $\sigma^{m-e}(C)\bm{y}^T=\bm{b}^T$
has a solution $\bm{y}=(y_1,y_2,\dots,y_{n})\in H^{n}$, where $\bm{y}=\sigma^{m-e}(\bm{x})$.
\end{itemize}
\end{lemma}

\begin{proof}
By the definition, we know that
$\widetilde{\Phi_{\bm{v}}(\C)}(G diag(\bm{v}),\infty)=\langle (\bm{c}_1\star \bm{v},0),\dots,(\bm{c}_{k-1}\star \bm{v},0),(\bm{c}_k\star \bm{v},1)\rangle$.
Then $\widetilde{\Phi_{\bm{v}}(\C)}(G diag(\bm{v}),\infty)$ is $e$-Galois self-orthogonal if and only if
$$\left\{
             \begin{array}{lr}
            \langle \bm{c}_i\star \bm{v},\bm{c}_j\star \bm{v}  \rangle_e=0,\ {\rm for}\ 0\leq i,\ j\leq k-1,\  {\rm except}\  i=j=k-1;&  \\
             \langle \bm{c}_{k+1}\star \bm{v},\bm{c}_{k+1}\star \bm{v}  \rangle_e=-1,&
             \end{array}
\right.$$
if and only if
$$\left\{
             \begin{array}{lr}
            \langle \bm{c}_i\star \sigma^e(\bm{c}_j), \bm{v}^{p^e+1}  \rangle_E=0,\ {\rm for }\ 0\leq i,\ j\leq k-1,\  {\rm except}\  i=j=k-1;&  \\
             \langle \bm{c}_{k+1}^{p^e+1},\bm{v}^{p^e+1}  \rangle_E=-1.&
             \end{array}
\right.$$
The rest of the proof is similar to Lemma \ref{lem galois oa}, and we omit the details.
This completes the proof.
\end{proof}

Note that $GRS_k(\bm{a},\bm{v},\infty)=\widetilde{GRS_k(\bm{a},\bm{v})}(G_k(\bm{a},\bm{v}),\infty)=
\widetilde{\Phi_{\bm{v}}(GRS_k(\bm{a},\bm{1}))}(G_k(\bm{a},\bm{1}) diag(\bm{v}),\infty)$ and
$G_k(\bm{a},\bm{1})=( \bm{a}^0,\bm{a}^1,\dots,\bm{a}^{k-1})^{\mathcal{T}}$.
By Lemma \ref{lem galois oa extend}, we can obtain the following lemma, which gives necessary and sufficient conditions for EGRS codes to be Galois self-orthogonal.

\begin{lemma}\label{lem <a,v> EG}
Let the two vectors $\bm{a}$ and $\bm{v}$ be defined as in Definition \ref{def GRS}.
Suppose that the matrix $A$ is defined as in Lemma \ref{lem <a,v>} and $\bm{b}=(0,\dots,0,-1)\in \F_q^{k^2}$.
Then the following statements are equivalent.
\begin{itemize}
\item[(1)] $GRS_k(\bm{a},\bm{v},\infty)$ is $e$-Galois self-orthogonal.
\item[(2)] $GRS_k(\bm{a},\bm{v},\infty)$ is $(m-e)$-Galois self-orthogonal.
\item[(3)] $\left\{
             \begin{array}{lr}
             \langle \bm{a}^{p^ei+j}, \bm{v}^{p^e+1} \rangle_E=0,\ 0\leq i,\ j\leq k-1,\  except\  i=j=k-1;&  \\
             \langle \bm{a}^{(p^e+1)(k-1)}, \bm{v}^{p^e+1} \rangle_E=-1.&
             \end{array}
\right.$
\item[(4)] $\left\{
             \begin{array}{lr}
             \langle \bm{a}^{p^{m-e}i+j}, \bm{v}^{p^{m-e}+1} \rangle_E=0,\ 0\leq i,\ j\leq k-1,\  except\  i=j=k-1;&  \\
             \langle \bm{a}^{(p^{m-e}+1)(k-1)}, \bm{v}^{p^{m-e}+1} \rangle_E=-1.&
             \end{array}
\right.$
\item[(5)] The equation $A\bm{x}^T=\bm{b}^T$ has a solution $\bm{x}=(x_1,x_2,\dots,x_{n})\in H^{n}$, where $\bm{x}=\bm{v}^{p^e+1}.$
\item[(6)] The equation $\sigma^{m-e}(A)\bm{y}^T=\bm{b}^T$
has a solution $\bm{y}=(y_1,y_2,\dots,y_{n})\in H^{n}$, where $\bm{y}=\sigma^{m-e}(\bm{x})$.
\end{itemize}
\end{lemma}

\subsection{Necessary and sufficient
conditions for (extended) GRS codes of length $n$ to be Galois self-orthogonal when $p^e\leq k\leq \lfloor \frac{n+p^e-1}{p^e+1}\rfloor$, where $2e\leq m$.}

In \cite{RefJ (2023) Yang FFA}, Li et al. gave some sufficient conditions for GRS and EGRS codes to be self-orthogonal. In this subsection, we prove that these conditions are also necessary when $p^e\leq k\leq \lfloor \frac{n+p^e-1}{p^e+1}\rfloor$, where $2e\leq m$.
Moreover, our results can be regarded as generalizations of Lemmas 2.2 and 2.3 in \cite{RefJ (2021) Fang It}, and be applied to construct new
Galois self-orthogonal MDS codes.

\begin{lemma}\label{lem yang 1}
(\cite[Lemma 2.4]{RefJ (2023) Yang FFA})
Let the two vectors $\bm{a}$ and $\bm{v}$ be defined as in Definition \ref{def GRS}.
If there exists a nonzero polynomial $\lambda(x)\in \F_q[x]$ with $\deg(\lambda(x))\leq n-(k-1)(p^e+1)-2$ such that
$\lambda(a_i)u_i=v_i^{p^e+1}\in H$, for $1\leq i\leq n$.
Then  $GRS_k(\bm{a},\bm{v})$ is $e$-Galois self-orthogonal.
\end{lemma}

\begin{lemma}\label{lem yang 2}
(\cite[Lemma 2.5]{RefJ (2023) Yang FFA})
Let the two vectors $\bm{a}$ and $\bm{v}$ be defined as in Definition \ref{def GRS}.
If there exists a monic polynomial $\lambda(x)\in \F_q[x]$
 with $\deg(\lambda(x))= n-(k-1)(p^e+1)-1$ such that
$-\lambda(a_i)u_i=v_i^{p^e+1}\in H$, for $1\leq i\leq n$.
Then $GRS_k(\bm{a},\bm{v},\infty)$ is $e$-Galois self-orthogonal.
\end{lemma}

\begin{lemma}\label{lem lam(x)}
Let the two vectors $\bm{a}$ and $\bm{v}$ be defined as in Definition \ref{def GRS}.
Assume $p^e\leq k\leq \lfloor \frac{n+p^e-1}{p^e+1}\rfloor$ and $2e\leq m$, then
$GRS_k(\bm{a},\bm{v})$ is $e$-Galois self-orthogonal if and only if
there exists a nonzero polynomial $\lambda(x)\in \F_q[x]$ with $\deg(\lambda(x))\leq n-(k-1)(p^e+1)-2$ such that
$\lambda(a_i)u_i=v_i^{p^e+1}\in H$, for $1\leq i\leq n$.
\end{lemma}

\begin{proof}
 By Lemma \ref{lem yang 1}, we need only prove necessity.
Since $k\geq p^e$, by Lemma \ref{lem <a,v>},
 $GRS_k(\bm{a},\bm{v})$ is $e$-Galois self-orthogonal
 if and only if the equation
 $$B\bm{x}^{T}=(\bm{a}^0,\bm{a}^1,\dots,\bm{a}^{(k-1)(p^e+1)})^{\mathcal{T}}\bm{x}^T=\bm{0}^T$$
has a solution $\bm{x}=\bm{v}^{p^e+1}=(x_1,x_2,\dots,x_n)\in H^n$.
Note that $rank(B)=(k-1)(p^e+1)+1$.
By \cite[Lemma 5]{RefJ (2008) Li},
the vectors
$$\bm{u}, \bm{a}\star \bm{u},\dots, \bm{a}^{n-(k-1)(p^e+1)-2}\star \bm{u}$$
form a basis of the solution space of $B\bm{x}^{T}=\bm{0}^T$.
Thus there exists $\lambda_l\in \F_q$ for $0\leq l\leq n-(k-1)(p^e+1)-2$ such that
$$x_i=\sum_{0\leq l\leq n-(k-1)(p^e+1)-2}\lambda_la_i^l u_i,\ {\rm for}\ 1\leq i\leq n,$$
Let $\lambda(x)=\sum_{l=0}^{n-(k-1)(p^e+1)-2}\lambda_lx^l$, then $\lambda(a_i)u_i=v_i^{p^e+1}$ for $1\leq i\leq n$.
This completes the proof.
\end{proof}

\begin{lemma}\label{lem lam(x) 2}
Let the two vectors $\bm{a}$ and $\bm{v}$ be defined as in Definition \ref{def GRS}.
Assume $p^e\leq k\leq \lfloor \frac{n+p^e}{p^e+1}\rfloor$ and $2e\leq m$, then
$GRS_k(\bm{a},\bm{v},\infty)$ is $e$-Galois self-orthogonal if and only if
there exists a monic polynomial $\lambda(x)\in \F_q[x]$ with $\deg(\lambda(x))= n-(k-1)(p^e+1)-1$ such that
$-\lambda(a_i)u_i=v_i^{p^e+1}\in H$, for $1\leq i\leq n$.
\end{lemma}

\begin{proof}
Similarly, by Lemma \ref{lem yang 2}, we need only prove necessity.
Since $k\geq p^e$, by Lemma \ref{lem <a,v> EG},
$GRS_k(\bm{a},\bm{v},\infty)$ is $e$-Galois self-orthogonal
 if and only if the equation\begin{equation}\label{eq EGRS 1}
B\bm{x}^{T}=(\bm{a}^0,\bm{a}^1,\dots,\bm{a}^{(k-1)(p^e+1)})^{\mathcal{T}}\bm{x}^T=\bm{b}^T
\end{equation}
has a solution $\bm{x}=\bm{v}^{p^e+1}=(x_1,x_2,\dots,x_n)\in H^n$,
where $\bm{b}=(0,\dots,0,-1)\in \F_q^{(k-1)(p^e+1)+1}$.
Similarly as in Lemma \ref{lem lam(x)}, from the first $(k-1)(p^e+1)$ lines of Eq. (\ref{eq EGRS 1}),
there exists $\lambda_l\in \F_q$ for $0\leq l\leq n-(k-1)(p^e+1)-1$ such that
$$x_i=\sum_{0\leq l\leq n-(k-1)(p^e+1)-1}\lambda_la_i^l u_i\ for\ 1\leq i\leq n.$$
Substituting $x_i$ to the last line of Eq. (\ref{eq EGRS 1}), we have
$$\lambda_{n-(k-1)(p^e+1)-1}\sum_{i=1}^n a_i^{n-1}u_i=-1,$$
which implies $\lambda_{n-(k-1)(p^e+1)-1}=-1$ from \cite[Lemma 5]{RefJ (2008) Li}.
Let
$$\lambda(x)=x^{n-(k-1)(p^e+1)-1}-\sum_{l=0}^{n-(k-1)(p^e+1)-2}\lambda_lx^l,$$
it follows that $-\lambda(a_i)u_i=v_i^{p^e+1}$ for $1\leq i\leq n$.
This completes the proof.
\end{proof}

 By Lemma \ref{lem lam(x)}, we can obtain the following relation between Galois self-orthogonal GRS codes.

\begin{lemma}\label{lem relation}
Let $\gcd(e,m)=\gcd(e',m)$. If there exists an $[n,k']_{q}$
$e'$-Galois self-orthogonal GRS code, for some $p^{e'}\leq k'\leq \lfloor \frac{n+p^{e'}-1}{p^{e'}+1}\rfloor$,
then for $1\leq k\leq \lfloor \frac{n+p^e-1-\deg(\lambda(x))}{p^e+1}\rfloor$, where $\deg(\lambda(x))\leq n-(k'-1)(p^{e'}+1)-2$, there exists an $[n,k]_q$ $e$-Galois self-orthogonal MDS code.
\end{lemma}

\begin{proof}
Note that $H=\langle w^{p^s+1}\rangle=\langle w^{p^e+1}\rangle=\langle w^{p^{e'}+1}\rangle$.
If there exists an $[n,k']_{q}$
$e'$-Galois self-orthogonal GRS code, by Lemma \ref{lem lam(x)},
there exists a nonzero polynomial $\lambda(x)\in \F_q[x]$ such that $\lambda(a_i)u_i\in H$, where $\deg(\lambda(x))\leq n-(k'-1)(p^{e'}+1)-2$.
Since $k\leq \lfloor \frac{n+p^e-1-\deg(\lambda(x))}{p^e+1}\rfloor$,
it follows that $\deg(\lambda(x))\leq n-(k-1)(p^e+1)-2$.
Hence we have $\lambda(a_i)u_i\in H$,
where $\deg(\lambda(x))\leq n-(k'-1)(p^s+1)-2$.
The result follows from Lemma \ref{lem lam(x)}.
 This completes the proof.
\end{proof}

By Lemma \ref{lem relation}, we have the following corollary.

\begin{corollary}\label{cor relation}
Suppose $\frac{m}{s}$ and $p$ are odd.
If there exists an $[n,k']_{q}$
Euclidean self-orthogonal GRS code,
then
for $1\leq k\leq 1+\lfloor \frac{2k'-2}{p^e+1}\rfloor$,
there exists an $[n,k]_q$ $e$-Galois self-orthogonal MDS code.
\end{corollary}

\begin{proof}
By Lemma \ref{lem relation},
it follows that
for $1\leq k\leq \lfloor \frac{n+p^e-1-\deg(\lambda(x))}{p^e+1}\rfloor$, where $\deg(\lambda(x))\leq n-2(k'-1)-2$, there exists an $[n,k]_q$ $e$-Galois self-orthogonal MDS code.
Note that
$$\lfloor \frac{n+p^e-1-\deg(\lambda(x))}{p^e+1}\rfloor\geq \lfloor \frac{p^e+1+2(k'-1)}{p^e+1}\rfloor=1+\lfloor \frac{2k'-2}{p^e+1}\rfloor,$$
by Corollary \ref{cor kkkk}, the required result can be obtained.
 This completes the proof.
\end{proof}

\begin{remark}
Corollary \ref{cor relation} generalises the previous results.
\begin{itemize}
\item
By \cite[Corollary 2.1]{RefJ (2023) Chen},
if there exists an $[n,k]_{q}$ $e$-Galois self-orthogonal GRS code, then there exists an $[n,k]_{q}$ MDS code with $l$-dimensional $e$-Galois hull, where $0\leq l\leq k$, where $q\geq 3$.
Therefore, if we take $e=s$ and $2k'=n$ in Corollary \ref{cor relation}, we obtain the results presented in \cite[Theorem 4.1]{RefJ (2022) Fang CC}.
\end{itemize}
\end{remark}

\section{Construction of Galois self-orthogonal MDS code}\label{sec4}

In this section, we use the methods introduced in Section \ref{sec3} to
construct some new classes of Galois self-orthogonal MDS codes.
 By Remark \ref{rem 1}, we only discuss the case $2e\leq m$.
 Depending on the length of the Galois self-orthogonal MDS codes, we divide this section into the following three subsections.

\subsection{$e$-Galois self-orthogonal MDS code of length $1\leq n\leq p^s+1$}

In this subsection, we construct some new classes of $e$-Galois self-orthogonal MDS code of length $1\leq n\leq p^{s}+1$.
It is worth noting that
the dimension of these $e$-Galois self-orthogonal MDS codes can be taken from $1$ to $\lfloor\frac{n}{2}\rfloor$.

\begin{lemma}\label{lem ax+b}
Let the two vectors $\bm{a}$ and $\bm{v}$ be defined as in Definition \ref{def GRS}.
If there exist $a\in \F_q^*$ and $b\in \F_q$ such that $\sigma^e(\bm{a})=a\bm{a}+b\bm{1}$ or $\sigma^{m-e}(\bm{a})=a\bm{a}+b\bm{1}$,
then the following statements hold.
\begin{itemize}
\item[(1)]
$GRS_k(\bm{a},\bm{v})$ is $e$-Galois self-orthogonal if and only if
there exists a nonzero polynomial $\lambda(x)\in \F_q[x]$ with $\deg(\lambda(x))\leq n-2k$ such that
$\lambda(a_i)u_i=v_i^{p^e+1}\in H$, for $1\leq i\leq n$.
\item[(2)]
$GRS_k(\bm{a},\bm{v},\infty)$ is $e$-Galois self-orthogonal if and only if
there exists a monic polynomial $\lambda(x)\in \F_q[x]$ with $\deg(\lambda(x))= n-2k+1$ such that
$-\lambda(a_i)u_i=v_i^{p^e+1}\in H$, for $1\leq i\leq n$.
\end{itemize}
\end{lemma}

\begin{proof}
By Lemma \ref{lem duichen}, we only need to prove the case $\sigma^e(\bm{a})=a\bm{a}+b\bm{1}$.

(1) By Lemma \ref{lem <a,v>},
$GRS_k(\bm{a},\bm{v})$ is $e$-Galois self-orthogonal
 if and only if the equation
 $A\bm{x}^T=\bm{0}^T$ has a solution $\bm{x}=\bm{v}^{p^e+1}=(x_1,x_2,\dots,x_n)\in H^n$,
 where the matrix $A$ be defined as in Lemma \ref{lem <a,v>}.
 Let
 $$B=(\bm{a}^0,\bm{a}^1,\dots,\bm{a}^{2k-2})^\mathcal{T}$$
be a $(2k-1)\times n$ matrix over $\F_{q}$.
Since $a\neq 0$ and
 \[\begin{split}
 a_l^{p^ei+j}&=(\sigma^e(a_l))^ia_l^j\\
 &=(aa_l+b)^ia_l^j,
\end{split}\]
where $1\leq l\leq n$ and $0\leq i,j\leq k-1$, all row vectors of $A$ are represented by all row vectors of $B$.
It is easy to check that the converse is also holds.
Then the matrix $A$ is row equivalent to the matrix $B$.
It follows that $GRS_k(\bm{a},\bm{v})$ is $e$-Galois self-orthogonal
 if and only if the equation
 $B\bm{x}^T=\bm{0}^T$ has a solution $\bm{x}=\bm{v}^{p^e+1}\in H^n$.
Similar to the proof of Lemma \ref{lem lam(x)}, we have that
the equation
 $B\bm{x}^T=\bm{0}^T$ has a solution $\bm{x}=\bm{v}^{p^e+1}\in H^n$
 if and only if
 there exists a nonzero polynomial $\lambda(x)\in \F_q[x]$ with $\deg(\lambda(x))\leq n-2k$ such that
$\lambda(a_i)u_i=v_i^{p^e+1}\in H$, for $1\leq i\leq n$.

(2) Similar to the proof of (1), we have that
 \[\begin{split}  & GRS_k(\bm{a},\bm{v},\infty)\ {\rm is}\ e\mbox{-}{\rm Galois\  self}\mbox{-}{\rm orthogonal}\\\
\Leftrightarrow\ &{\rm the\ equation\ }
 A\bm{x}^T=\bm{b}^T\ {\rm has\ a\ solution}\ \bm{x}=\bm{v}^{p^e+1}\in H^n, \\
 \Leftrightarrow\ &{\rm the\ equation\ }
 B\bm{x}^T=\tilde{\bm{b}}^T\ {\rm has\ a\ solution}\ \bm{x}=\bm{v}^{p^e+1}\in H^n, \\
\Leftrightarrow\ & {\rm there\ exists\ a\ monic\ polynomial}\ \lambda(x)\in \F_q[x]\ {\rm with}\ \deg(\lambda(x))= n-2k+1\\
 & {\rm such\ that}\ -\lambda(a_i)u_i=v_i^{p^e+1}\in H,\ {\rm for}\ 1\leq i\leq n,\\
\end{split}\]
where $\bm{b}=(0,\dots,0,-1)\in \F_q^{k^2}$ and $\tilde{\bm{b}}=(0,\dots,0,-1)\in \F_q^{2k-1}$.
 This completes the proof.
\end{proof}

\begin{remark}
If we take $e=0$, $a=1$ and $b=0$ in Lemma \ref{lem ax+b}, we obtain the results presented in  \cite[Lemmas 2.2 and 2.3]{RefJ (2021) Fang It}.
In other words, Lemma \ref{lem ax+b} is a generalization of \cite[Lemmas 2.2 and 2.3]{RefJ (2021) Fang It}.
\end{remark}

\begin{theorem}\label{th 1 ax+b}
If there exist $n$ distinct elements $a_1,a_2,\dots,a_{n}\in \F_q$ such that $\sigma^e(a_i)=aa_i+b$
or $\sigma^{m-e}(a_i)=aa_i+b$ $(1\leq i\leq n)$ for some $a\in \F_q^*$ and $b\in \F_q$, then the following statements hold.
\begin{itemize}
\item[(1)]
Suppose $\frac{m}{s}$ is odd and $p$ is even, then for $1\leq k\leq \lfloor \frac{n}{2}\rfloor$,
there exists an $[n,k]_q$ $e$-Galois self-orthogonal MDS code.
\item[(2)]
Suppose $\frac{m}{s}$ is odd and $p$ is even, then
for $n<q$ and $1\leq k\leq \lfloor \frac{n+1}{2}\rfloor$,
there exists an $[n+1,k]_q$ $e$-Galois self-orthogonal MDS code.
\item[(3)]
Suppose $\frac{m}{s}$ and $p$ are odd.
If there exists an $[n,k]_{q}$
Euclidean self-orthogonal (extended) GRS code with code locators $a_1,a_2,\dots,a_n$,
then there exists an $[n,k]_q$ $e$-Galois self-orthogonal MDS code.
\end{itemize}
\end{theorem}

\begin{proof}
Similarly,
we only need to prove the case $\sigma^e(a_i)=aa_i+b$, $1\leq i\leq n$.

Let $\bm{a}=(a_1,a_2,\dots,a_n)\in \F_q^n$,
it follows that $\sigma^e(\bm{a})=a\bm{a}+b\bm{1}$.

(1)
For $1\leq k\leq \lfloor\frac{n}{2}\rfloor$,
let $\lambda(x)=1$ with $\deg(\lambda(x))=0\leq n-2k$.
Since $\frac{m}{s}$ is odd and is $p$ even, by Lemma \ref{lem HHH}, we have $H=\F_q^*$.
It follows that $\lambda(a_i)u_i=u_i\in H=\F_q^*$, for $1\leq i\leq n$.
Then by Lemma \ref{lem ax+b},
there exists a vector $\bm{v}\in (\F_{q}^*)^n$
such that $GRS_{k}(\bm{a},\bm{v})$ is $e$-Galois self-orthogonal.

(2)
Since $n<q$, there exists $c\in \F_q$ such that $\bm{a}+c\bm{1}\in (\F_q^*)^n$.
Note that $\sigma^e(\bm{a}+c\bm{1})=a\bm{a}+(b+\sigma^e(c))\bm{1}$.
Then we always assume that $\bm{a}\in (\F_q^*)^n$ without loss of generality.
For $1\leq k\leq \lfloor\frac{n+1}{2}\rfloor$, let $\lambda(x)=x^{n-2k+1}$ with $\deg(\lambda(x))=n-2k+1$.
It follows that $-\lambda(a_i)u_i=-a_i^{n-2k+1}u_i\in H=\F_q^*$, for $1\leq i\leq n$.
Then by Lemma \ref{lem ax+b},
there exists a vector $\bm{v}\in (\F_{q}^*)^n$
such that $GRS_{k}(\bm{a},\bm{v},\infty)$ is $e$-Galois self-orthogonal.

(3)
We only prove the case for GRS codes, which can be obtained similarly for EGRS codes.
Since $\frac{m}{s}$ and $p$ are odd, by Lemma \ref{lem HHH}, we have $H=QR$.
If there exists an $[n,k]_{q}$ Euclidean self-orthogonal GRS code $GRS_k(\bm{a},\bm{v})$ with $\sigma^e(\bm{a})=a\bm{a}+b\bm{1}$,
 then by Lemma \ref{lem ax+b}, there exists a nonzero polynomial $\lambda(x)\in \F_q[x]$ with $\deg(\lambda(x))\leq n-2k$ such that
$\lambda(a_i)u_i\in QR=H$, for $1\leq i\leq n$.
The desired result follows from Lemma \ref{lem ax+b}.
This completes the proof.
\end{proof}

When $a=1$ and $b=0$,
we can take $n$ distinct elements $a_1,a_2,\dots,a_n$ from $\F_{p^{s}}$ such that
$\sigma^e(a_i)=a_i$ for $1\leq i\leq n$.
By Theorem \ref{th 1 ax+b} and the next Theorem \ref{th p2s 1}, we have the following corollary.

\begin{corollary}\label{cor ps 2}
\begin{itemize}
\item[(1)]
Suppose $\frac{m}{s}$ is odd and $p$ is even,
then for $1\leq n\leq   \min\{p^{s}+1,q\}$ and $1\leq k\leq \lfloor \frac{n}{2}\rfloor$, there exists an $[n,k]_q$ $e$-Galois self-orthogonal MDS code.

\item[(2)]
Suppose $\frac{m}{s}$ and $p$ are odd.
For $1\leq n\leq p^{s}+1$.
If there exists an $[n,k]_{p^{s}}$
Euclidean self-orthogonal (extended) GRS code,
then there exists an $[n,k]_q$ $e$-Galois self-orthogonal MDS code.

\item[(3)]
Suppose $\frac{m}{s}$ is even,
then for $1\leq n\leq p^{s}+1$ and $1\leq k\leq \lfloor \frac{n}{2}\rfloor$, there exists an $[n,k]_q$ $e$-Galois self-orthogonal MDS code.
\end{itemize}
\end{corollary}

\begin{proof}
In \cite{RefJ (2004) M. R,RefJ (2004) M. G}, all possible parameters of $p^{2s}$-ary
Hermitian self-orthogonal (extended) GRS codes of length $1\leq n\leq p^s+1$ are completely determined.
Then the proof of Corollary \ref{cor ps 2} (3) follows from the next Theorem \ref{th p2s 1} (3).
This completes the proof.
\end{proof}

\begin{example}
In this example,
we give some $e$-Galois self-orthogonal MDS codes from Corollary \ref{cor ps 2}.
Although our length is small compared to the field, the dimension can be taken to the maximum.
In particular, when $n$ is even, we can obtain some $e$-Galois self-dual MDS codes.
\begin{itemize}
\item[(1)]
Take $(p,m)=(2,3s)$ in Corollary \ref{cor ps 2} (1), then all the cases that satisfy that $\frac{m}{s}$ is odd are $e=s$ and $e=2s$.
Then we obtain some $[n,k]_{p^{3s}}$ $e$-Galois self-orthogonal MDS codes
 with the following parameters.
 \begin{center}
\begin{tabular}{ccc}
\hline
$e$ & $n$ & $k$ \\
\hline
$\{s,2s\}$ & $1\leq n\leq p^s+1$ &  $1\leq k\leq \lfloor \frac{n}{2}\rfloor$ \\
\hline
\end{tabular}
\end{center}

\item[(2)]
Take $m=3s$, $s$ even and $p$ odd in Corollary \ref{cor ps 2} (2).
By Lemma \ref{lem rule} (2), if there exists an $[n,\frac{n}{2}]_{p^s}$ Euclidean self-dual (extended) GRS code,
where $n$ is even, then there exists an $[n-1,\frac{n}{2}-1]_{p^s}$ Euclidean self-orthogonal (extended) GRS code.
From the results in \cite{RefJ (2020) Feng MDS it} and \cite{RefJ (2021) Ning MDS it}, there exist Euclidean self-orthogonal (extended) GRS codes with parameters $[n,k]_{p^s}$ for $1\leq n\leq 3p^{\frac{s}{2}}-3$ and $1\leq k\leq \lfloor\frac{n}{2}\rfloor$.
Then we obtain some $[n,k]_{p^{3s}}$ $e$-Galois self-orthogonal MDS codes
 with the following parameters.
 \begin{center}
\begin{tabular}{ccc}
\hline
$e$ & $n$ & $k$ \\
\hline
$\{s,2s\}$ & $1\leq n\leq 3p^{\frac{s}{2}}-3$ &  $1\leq k\leq \lfloor \frac{n}{2}\rfloor$ \\
\hline
\end{tabular}
\end{center}

\item[(3)]
Take $m=4s$ in Corollary \ref{cor ps 2} (3), then all the cases that satisfy that $\frac{m}{s}$ is even are $e=s$ and $e=3s$.
Then we obtain some $[n,k]_{p^{4s}}$ $e$-Galois self-orthogonal MDS codes
 with the following parameters.
 \begin{center}
\begin{tabular}{ccc}
\hline
$e$ & $n$ & $k$ \\
\hline
$\{s,3s\}$ & $1\leq n\leq p^s+1$ &  $1\leq k\leq \lfloor \frac{n}{2}\rfloor$ \\
\hline
\end{tabular}
\end{center}
\end{itemize}
\end{example}

\subsection{$e$-Galois self-orthogonal MDS code of length $p^s+1\leq n\leq p^{2s}+1$}

In this subsection, we construct some new classes of $e$-Galois self-orthogonal MDS code of length $p^s+1\leq n\leq p^{2s}+1$.
The $e$-Galois self-orthogonal MDS code we construct can have dimensions greater than $\lfloor \frac{n+p^e-1}{p^e+1}\rfloor$.

\begin{theorem}\label{th p2s 1}
\begin{itemize}
\item[(1)]
Suppose $\frac{m}{s}$ is odd and $p$ is even.
For $p^{s}+2\leq n\leq p^{2s}$,
let $n=\sum_{i=1}^tn_i$ for some positive integers $n_i\leq p^{s}$ and $t\leq p^{s}$.
Then for $1\leq k\leq \max\{\lfloor \frac{n+p^e-1}{p^e+1}\rfloor, \min\{\lfloor \frac{n_1}{2}\rfloor,\lfloor \frac{n_2}{2}\rfloor,\dots,\lfloor \frac{n_t}{2}\rfloor \} \}$,
 there exists an $[n,k]_q$ $e$-Galois self-orthogonal MDS code.

\item[(2)]
Suppose $\frac{m}{s}$ and $p$ are odd.
For $p^{s}+2\leq n\leq p^{2s}$,
let $n=\sum_{i=1}^tn_i$ for some positive integers $n_i\leq p^{s}$ and $t\leq p^{s}$.
For $1\leq i\leq t$,
if there exists an $[n_i,k_i]_{p^{s}}$
Euclidean self-orthogonal GRS code,
then for $1\leq k\leq \min\{k_1,k_2,\dots,k_t\} $,
 there exists an $[n,k]_q$ $e$-Galois self-orthogonal MDS code.

\item[(3)]
Suppose $\frac{m}{s}$ is even.
For $1\leq n\leq p^{2s}+1$, if there exists an $[n,k]_{p^{2s}}$
Hermitian self-orthogonal (extended) GRS code,
then there exists an $[n,k]_q$ $e$-Galois self-orthogonal MDS code.
\end{itemize}

\end{theorem}

\begin{proof}
Label the elements of $\F_{p^s}$ as $\{0=b_1,b_2,\dots,b_{p^s}\}$.
Choose $\zeta\in \F_{q}\setminus \F_{p^s}$.
For $1\leq i\leq p^s$,
define
$$S_i=\F_{p^s}+b_i\zeta=\{s+b_i\zeta: s\in \F_{p^s}\}.$$
It is easy to check that $S_i\cap S_j=\emptyset$ for $i\neq j$ and $|S_i|=p^s$.

(1) On one hand,
since $n_i\leq p^{s}$ and $t\leq p^{s}$, let
$\bm{a}_i=(a_{i1},a_{i2},\dots,a_{in_i})\in (S_i)^{n_i}$ for $1\leq i\leq t$.
By Corollary \ref{cor ps 2} (1) and Lemma \ref{lem equiv}, there exists a vector $\bm{v}_i=(v_{i1},v_{i2},\dots,v_{in_i})\in (\F_{q}^*)^{n_i}$
such that $GRS_{k_i}(\bm{a}_i,\bm{v}_i)$ is an $e$-Galois self-orthogonal GRS code for
$1\leq i\leq t$ and $1\leq k_i\leq \lfloor \frac{n_i}{2} \rfloor$.
Let $\bm{a}=(\bm{a}_1,\bm{a}_2,\dots,\bm{a}_t)\in \F_q^n$ and $\bm{v}=(\bm{v}_1,\bm{v}_2,\dots,\bm{v}_t)\in (\F_q^*)^n$,
it follows that
$$\langle \bm{a}^{p^ei+j}, \bm{v}^{p^e+1} \rangle_E=\sum_{i=1}^{t}\langle \bm{a}_i^{p^ei+j}, \bm{v}_i^{p^e+1} \rangle_E=0,$$
for $0\leq i,j\leq k-1\leq \min\{\lfloor \frac{n_1}{2}\rfloor,\lfloor \frac{n_2}{2}\rfloor,\dots,\lfloor \frac{n_t}{2}\rfloor \}-1$.
The result follows from
Lemma \ref{lem <a,v>}.

On the other hand,
for $p^s+1 \leq n\leq q$, let $\bm{a}=(a_1,a_2,\dots,a_n)\in \F_{q}^n$,
where $a_1, a_2,\dots, a_n$ are distinct elements in $\F_{q}$.
Since $1\leq k \leq \lfloor \frac{n+p^e-1}{p^e+1}\rfloor$,
it follows that $n-(k-1)(p^e+1)-2\geq 0$.
Let $\lambda(x)=1$ with $\deg(\lambda(x))=0\leq n-(k-1)(p^e+1)-2$.
Then we have
$\lambda(a_i)u_i=u_i\in\F_q^*=H,\ 1\leq i\leq n$.
The result follows from
Lemma \ref{lem yang 1}.

(2)
 Since $n_i\leq p^{s}$ and $t\leq p^{s}$, let
$\bm{a}_i=(a_{i1},a_{i2},\dots,a_{in_i})\in (S_i)^{n_i}$ for $1\leq i\leq t$.
By Corollary \ref{cor ps 2} (2) and Lemma \ref{lem equiv}, there exists a vector $\bm{v}_i=(v_{i1},v_{i2},\dots,v_{in_i})\in (\F_{q}^*)^{n_i}$
such that $GRS_{k_i}(\bm{a}_i,\bm{v}_i)$ is an $e$-Galois self-orthogonal GRS code for
$1\leq i\leq t$.
The rest of the proof is similar to (1), which we omit.

(3)
We only prove the case for GRS codes, which can be obtained similarly for EGRS codes.
Since $2s\mid m$, then $\F_{p^{2s}}\subseteq \F_{q}$.
If there exists an $[n,k]_{p^{2s}}$
Hermitian self-orthogonal GRS code, by Lemma \ref{lem <a,v>}, then there exist vectors $\bm{a}\in \F_{p^{2s}}^n$, where $a_1, a_2,\dots, a_n$ are distinct elements and $\bm{v}\in (\F_{p^{2s}}^*)^n$ such that $\langle \bm{a}^{p^si+j}, \bm{v}^{p^s+1} \rangle_E=0$, for all $0\leq i,j \leq k-1$.
In the following we prove that $2s\nmid e$.
If $2s\mid e$, since $2s\mid m$, then $2s\mid \gcd(e,m)=s$.
Clearly, this is a contradiction.
Then we have $2s\nmid e$,
let $\frac{e}{s}=2\tau+1$, for some $\tau \geq 1$.
Note that $\bm{a}^{p^e}=\bm{a}^{p^{2s\tau+s}}=(\bm{a}^{p^{2s\tau}})^{p^s}=\bm{a}^{p^s}$
and $\bm{v}^{p^e}=\bm{v}^{p^s}$, then we have
$$\langle \bm{a}^{p^ei+j}, \bm{v}^{p^e+1} \rangle_E=\langle \bm{a}^{p^si+j}, \bm{v}^{p^s+1} \rangle_E=0,$$
 for all $0\leq i,j \leq k-1$.
The result follows from
Lemma \ref{lem <a,v>}.
 This completes the proof.
\end{proof}

\begin{example}
In this example,
we give some $e$-Galois self-orthogonal MDS codes from Theorem \ref{th p2s 1}.
\begin{itemize}
\item[(1)]
In Theorem \ref{th p2s 1} (1), taking $(p,m)=(2,3s)$, we obtain some $[n,k]_{p^{3s}}$ $e$-Galois self-orthogonal MDS codes
 with the following parameters.
\begin{center}
\begin{tabular}{ccc}
\hline
$e$ & $n$ & $k$ \\
\hline
$\{s,2s\}$ & $rp^s$, $2\leq r\leq p^s$ &  $1\leq k\leq \frac{p^s}{2}$ \\
$\{s,2s\}$ & $(p^s+1)r+2\leq n\leq (p^s+1)(r+1)+1$, $1\leq r\leq p^s-2$ &  $1\leq k\leq r+1$ \\
\hline
\end{tabular}
\end{center}

\item[(2)]
Take $m=3s$ and $p$ odd in Theorem \ref{th p2s 1} (2).
It is easy to see that there exist Euclidean self-orthogonal GRS codes with parameters $[p^s,k]_{p^s}$ for $1\leq k\leq \lfloor \frac{p^s}{2}\rfloor$.
Then we obtain some $[n,k]_{p^{3s}}$ $e$-Galois self-orthogonal MDS codes
 with the following parameters.
 \begin{center}
\begin{tabular}{ccc}
\hline
$e$ & $n$ & $k$ \\
\hline
$\{s,2s\}$ & $rp^s$, $2\leq r\leq p^s$ &  $1\leq k\leq \lfloor \frac{p^s}{2}\rfloor$ \\
\hline
\end{tabular}
\end{center}

\item[(3)]
Take $m=4s$ and $p$ odd in Theorem \ref{th p2s 1} (3).
From the known results of Hermitian self-orthogonal (extended) GRS codes, it follows that there exist Hermitian self-orthogonal (extended) GRS codes with parameters $[n,k]_{p^{2s}}$ for
$n\in \{p^{2s},p^{2s}+1,\frac{p^{2s}+1}{2},p^{2s}-1,\frac{p^{2s}-1}{2}\}$ and $1\leq k\leq p^s-1$.
Then we obtain some $[n,k]_{p^{4s}}$ $e$-Galois self-orthogonal MDS codes
 with the following parameters.
 \begin{center}
\begin{tabular}{ccc}
\hline
$e$ & $n$ & $k$ \\
\hline
$\{s,3s\}$ & $\{p^{2s},p^{2s}+1,\frac{p^{2s}+1}{2},p^{2s}-1,\frac{p^{2s}-1}{2}\}$ &  $1\leq k\leq p^s-1$ \\
\hline
\end{tabular}
\end{center}

\end{itemize}
\end{example}

In the following, we give a specific class of $e$-Galois self-orthogonal MDS codes.
Although the method is similar to the current construction of Hermitian self-orthogonal GRS codes,
it is shown that the method is feasible for $e$-Galois self-orthogonal GRS codes when $\frac{m}{e}$ is odd and $p$ is even.
Moreover, new $e$-Galois self-orthogonal MDS codes with dimension greater than $\lfloor \frac{n+p^e-1}{p^e+1}\rfloor$ are obtained.

For integers $s_1$ and $s_2$, we write
\[[s_1,s_2]=\begin{cases}
\emptyset, & if\ s_1>s_2; \\
\{s_1,s_1+1,\dots,s_2\}, & if\  s_1\leq s_2.
\end{cases}\]

\begin{lemma}\label{lem zhenchu}
Suppose $1\leq k\leq t\leq p^e-1$.
Then for any $0\leq i,j\leq k-1$, $p^ei+j=u(p^e-1)$ if and only if
$u\in\{0\}\cup [p^e-t+1,t]$.
\end{lemma}

\begin{theorem}\label{th 11 3333}
Suppose $\frac{m}{e}$ is odd and $p$ is even.
Let $n=r(p^e-1)+1$, where $1\leq r\leq p^e+1$, then for $1\leq k \leq \lfloor \frac{p^e+r-1}{2}\rfloor$, there exists an $[n,k]_q$ $e$-Galois self-orthogonal MDS code.
\end{theorem}

 The proofs of Lemma \ref{lem zhenchu} and Theorem \ref{th 11 3333} are similar to the construction of Hermitian self-orthogonal GRS codes, which we omit here, and the details are given in the Appendix.

\begin{example}
In this example,
we give some Galois self-orthogonal MDS codes from Theorem \ref{th 11 3333}.

\begin{itemize}
\item
In Theorem \ref{th 11 3333}, taking $m=3e$, we obtain some $[n,k]_{p^{3e}}$ $e$-Galois self-orthogonal MDS codes
 with the following parameters.
\begin{center}
\begin{tabular}{ccc}
\hline
$e$ & $n$ & $k$ \\
\hline
$\{e,2e\}$ & $r(p^e-1)+1$, $1\leq r\leq p^e+1$ &  $1\leq k\leq \lfloor \frac{p^e+r-1}{2}\rfloor$ \\
\hline
\end{tabular}
\end{center}

\end{itemize}
\end{example}

\subsection{$e$-Galois self-orthogonal MDS code of length $n=q+1$}

In this subsection, we construct $q$-ary $e$-Galois self-orthogonal MDS codes of length $q+1$.
The dimension of our $e$-Galois self-orthogonal EGRS code is close to or equal to the value $\lceil\frac{q+p^e+1}{p^e+1}\rceil-1$.
 For the need of the proof, we introduce the following lemma.

\begin{lemma}\label{lem m(x)}
(\cite[Lemma 12]{RefJ (2018) Fang two})
For any integer $l\geq 2$ and finite field $\F_q$, there exists a monic polynomial
$m(x)\in \F_{q}[x]$ of degree $l$ such that $m(a)\neq 0$ for all $a\in \F_q$.
\end{lemma}

\begin{theorem}\label{th 2 q+1}
Let $n=q+1$, then the following statements hold.
\begin{itemize}
\item[(1)]
If  $\frac{m}{s}$ is odd and $p$ is even, then for $1\leq k \leq \lfloor \frac{q+p^e-2}{p^e+1}\rfloor$, there exists an $[n,k]_q$ $e$-Galois self-orthogonal MDS code.
\item[(2)]
If $\frac{m}{s}$ and $p$ are odd,
then for $1\leq k\leq \lfloor \frac{q+p^e-4}{p^e+1}\rfloor$, there exists an $[n,k]_q$ $e$-Galois self-orthogonal MDS code.
\item[(3)]
If $\frac{m}{s}$ is even,
then for $1\leq k\leq \lfloor \frac{q+p^e-2(p^s+1)}{p^e+1}\rfloor$ or $k=\frac{q-1}{p^e+1}+1$ $(s=e)$,
there exists an $[n,k]_q$ $e$-Galois self-orthogonal MDS code.
\end{itemize}
\end{theorem}

\begin{proof}
Let $\bm{a}=(a_1,a_2,\dots,a_{q})\in \F_{q}^q$,
where $a_1, a_2,\dots, a_q$ are distinct elements in $\F_{q}$.
Then we have $u_i=\prod_{a\in \F_{q}^*}a=-1$ for $1\leq i\leq q$.
Denote $l=q-(k-1)(p^e+1)-1$.

(1)
Since $1\leq k \leq \lfloor \frac{q+p^e-2}{p^e+1}\rfloor$, we have $l=q-(k-1)(p^e+1)-1\geq 2$.
By Lemma \ref{lem m(x)},
there exists a monic polynomial
$m(x)\in \F_{q}[x]$ of degree $l$ such that $m(a)\neq 0$ for all $a\in \F_q$.
Let $\lambda(x)=m(x)$, then
$\deg(\lambda(x))=l$ and
$\lambda(a_i)\in \F_q^*$
for $1\leq i\leq q$.
It implies that
$-\lambda(a_i)u_i=\lambda(a_i)\in\F_q^*$ for $1\leq i\leq q$.
Then by Lemma \ref{lem yang 2},
there exists a vector $\bm{v}\in (\F_{q}^*)^n$
such that $GRS_{k}(\bm{a},\bm{v},\infty)$ is $e$-Galois self-orthogonal.

(2)
Since $1\leq k \leq \lfloor \frac{q+p^e-4}{p^e+1}\rfloor$,
we have $l=q-(k-1)(p^e+1)-1\geq 4$.
Note that $\gcd(p^e+1,q-1)=2$,
it follows that $2\mid l$ and $\frac{l}{2}\geq 2$.
By Lemma \ref{lem m(x)},
there exists a monic polynomial
$m(x)\in \F_{q}[x]$ of degree $\frac{l}{2}$ such that $m(a)\neq 0$ for all $a\in \F_q$.
Let $\lambda(x)=m(x)^2$, then
$\deg(\lambda(x))=l$ and
$\lambda(a_i)\in QR$
for $1\leq i\leq q$.
It implies that
$-\lambda(a_i)u_i=m(a_i)^2\in QR$ for $1\leq i\leq q$.
The desired result follows from Lemma \ref{lem yang 2}.

(3)
On one hand,
if $1\leq k \leq \lfloor \frac{q+p^e-2(p^s+1)}{p^e+1}\rfloor$,
we have $l=q-(k-1)(p^e+1)-1\geq 2(p^s+1)$.
Note that $\gcd(p^e+1,q-1)=p^s+1$,
it follows that $(p^s+1)\mid l$ and $\frac{l}{p^s+1}\geq 2$.
By Lemma \ref{lem m(x)},
there exists a monic polynomial
$m(x)\in \F_{q}[x]$ of degree $\frac{l}{p^s+1}$ such that $m(a)\neq 0$ for all $a\in \F_q$.
Let $\lambda(x)=m(x)^{p^s+1}$, then
$\deg(\lambda(x))=l$ and
$\lambda(a_i)\in \langle w^{p^s+1}\rangle$
for $1\leq i\leq q$.
It implies that
$-\lambda(a_i)u_i=m(a_i)^{p^s+1}\in \langle w^{p^s+1}\rangle$ for $1\leq i\leq n$.
The desired result follows from Lemma \ref{lem yang 2}.

On the other hand,
since $s=e$,
we have $(p^e+1)\mid (q-1)$.
If $k=\frac{q-1}{p^e+1}+1$,
 it follows that $l=0$.
Let $\lambda(x)=x^0=1$,
then
$-\lambda(a_i)u_i=1\in \langle w^{p^s+1}\rangle$ for $1\leq i\leq n$.
The desired result follows from Lemma \ref{lem yang 2}.
 This completes the proof.
\end{proof}

\begin{example}
In this example,
we give some $e$-Galois self-orthogonal MDS codes from Theorem \ref{th 2 q+1}.
\begin{itemize}
\item[(1)]
In Theorem \ref{th 2 q+1} (1), taking $(p,m,s)=(2,7,1)$, we obtain some $[n,k]_{2^{7}}$ $e$-Galois self-orthogonal MDS codes
 with the following parameters.
\begin{center}
\begin{tabular}{ccc| c c c| c c c}
\hline
$e$ & $n$ & $k$ & $e$ & $n$ & $k$ & $e$ & $n$ & $k$\\
\hline
$\{1,6\}$ & $129$ &  $1\leq k\leq 42$ &$\{2,5\}$ & $129$ & $1\leq k\leq 26$&$\{3,4\}$ & $129$ & $1\leq k\leq 14$\\
\hline
\end{tabular}
\end{center}

\item[(2)]
In Theorem \ref{th 2 q+1} (2), taking $(p,m,s)=(3,5,1)$, we obtain some $[n,k]_{3^{5}}$ $e$-Galois self-orthogonal MDS codes
 with the following parameters.
\begin{center}
\begin{tabular}{ccc| c c c}
\hline
$e$ & $n$ & $k$ & $e$ & $n$ & $k$ \\
\hline
$\{1,4\}$ & $244$ &  $1\leq k\leq 60$ &$\{2,3\}$ & $244$ & $1\leq k\leq 24$\\
\hline
\end{tabular}
\end{center}
\item[(3)]
In Theorem \ref{th 2 q+1} (3), taking $(p,m,s)=(3,8,2)$ or $(p,m,s)=(3,8,1)$, we obtain some $[n,k]_{3^{8}}$ $e$-Galois self-orthogonal MDS codes
 with the following parameters.
\begin{center}
\begin{tabular}{ccc| c c c}
\hline
$e$ & $n$ & $k$ & $e$ & $n$ & $k$ \\
\hline
$\{1,7\}$ & $6562$ & $1\leq k\leq 1639$ or $k=1641$ &$\{2,6\}$ & $6562$ & $1\leq k\leq 655$ or $k=657$ \\
$3$ & $6562$ & $1\leq k\leq 235$ &$5$ & $6562$ & $1\leq k\leq 235$ \\
\hline
\end{tabular}
\end{center}
\end{itemize}
\end{example}

\section{Comparison of results and new MDS codes with $e$-Galois hulls}\label{sec5}

In this section, we compare the Galois self-orthogonal MDS codes in this paper with known Galois self-orthogonal MDS codes. Moreover, using the propagation rules, we obtain some new MDS codes with Galois hulls.
As an application, many quantum codes can be obtained from these MDS codes with Galois
hulls.

\subsection{Comparison of results}

In Table \ref{tab:1}, we list the parameters of the known Galois self-orthogonal MDS codes.
 From Table \ref{tab:1}, we can find that most of the known $e$-Galois self-orthogonal MDS codes have dimensions less than $\lfloor \frac{n+p^e-1}{p^e+1}\rfloor$ or $\lfloor \frac{n+p^{m-e}-1}{p^{m-e}+1}\rfloor$.
By Remark \ref{rem 1},
we only consider the case $2e\leq m$ without loss of generality.
Note that if $2e\leq m$, we have $\lfloor \frac{n+p^{m-e}-1}{p^{m-e}+1}\rfloor\leq \lfloor \frac{n+p^e-1}{p^e+1}\rfloor$. Thus we need only consider a bound $\lfloor \frac{n+p^e-1}{p^e+1}\rfloor$.

In Table \ref{tab:2}, we list our constructions of Galois self-orthogonal MDS codes.
It can be found that the dimension of the Galois self-orthogonal MDS code we constructed can be greater than $\lfloor \frac{n+p^e-1}{p^e+1}\rfloor$, which is different from most of the codes in Table \ref{tab:1}.
Moreover, the Galois self-orthogonal MDS codes we constructed have more flexible lengths compared to the results of Islam et al \cite{RefJ (2023) Islam}.
In addition, we also consider Galois self-orthogonal MDS codes of length $q+1$, which are not available in Table \ref{tab:1}.
Therefore, the Galois self-orthogonal MDS code we construed has new parameters.

\subsection{MDS codes with Galois hulls of arbitrary dimensions}

From the following propagation rules, we can obtain many MDS codes with Galois hulls of arbitrary dimensions from a known $e$-Galois self-orthogonal (extended) GRS code.

\begin{lemma}\label{lem rule}(propagation rule)
Let $q>4$ be a prime power. If there exists an $[n,k]_{q}$ $e$-Galois self-orthogonal (extended) GRS code, then the following statements hold.
\begin{itemize}
\item[(1)]
(\cite[Corollary 2.1]{RefJ (2023) Chen}) There exists an $[n,k]_{q}$ MDS code with $l$-dimensional $e$-Galois hull for each $0\leq l\leq k$;
\item[(2)]
(Combining \cite[Corollary 2.13]{RefJ (2023) Yang FFA} and \cite[Corollary 2.1]{RefJ (2023) Chen})
For $1\leq i\leq k$, there exists an $[n-i,k-i]_{q}$ MDS code with $l$-dimensional $e$-Galois hull for each $0\leq l\leq k-i$;
\item[(3)]
(Combining \cite[Corollary 2.14]{RefJ (2023) Yang FFA} and \cite[Corollary 1]{RefJ (2023) Wan 3})
For $1\leq i\leq k$, there exists an $[n-i,k]_{q}$ MDS code with $l$-dimensional $e$-Galois hull for each $0\leq l\leq k-i$;
\item[(4)]
(\cite[Corollary 3]{RefJ (2023) Wan 3})
For $1\leq i\leq \min\{k,q+1-n\}$, there exists an $[n+i,k]_{q}$ MDS code with $l$-dimensional $e$-Galois hull for each $0\leq l\leq k-i$;
\item[(5)]
(\cite[Corollary 5]{RefJ (2023) Wan 3})
If $n<q+1$, for $1\leq i\leq k$, there exists an $[n,k+i]_{q}$ MDS code with $l$-dimensional $e$-Galois hull for each $0\leq l\leq k-i$;
\item[(6)]
(\cite[Corollary 9]{RefJ (2023) Wan 3})
If $n<q+1$, for $1\leq i\leq \min\{k,q+1-n\}$, there exists an $[n+i,k+i]_{q}$ MDS code with $l$-dimensional $e$-Galois hull for each $0\leq l\leq k-i$;
\item[(7)]
(Combining \cite[Corollary 2.1]{RefJ (2023) Chen} and
Corollary \ref{cor kkkk})
If $n< q+1$, for $1\leq i\leq k$, there exists an $[n,k-i]_q$ MDS code with $l$-dimensional $e$-Galois hull for each $0\leq l\leq k-i$.
\end{itemize}
\end{lemma}

 Note that the dimension of the Galois self-orthogonal MDS codes we obtained can be larger than $\lfloor \frac{n+p^e-1}{p^e+1}\rfloor$, then the dimension of the MDS codes with Galois hulls of arbitrary dimensions obtained by propagation rules can also be larger than $\lfloor \frac{n+p^e-1}{p^e+1}\rfloor$.
 Therefore,
 from the Galois self-orthogonal MDS codes constructed in this paper and Lemma \ref{lem rule}, we can directly obtain many new MDS codes with Galois hulls of arbitrary dimensions.

\subsection{Application to quantum codes}

We always use the notation $[[n,k,d]]_q$ to denote a $q$-ary QECC of length $n$, dimension $q^k$ and minimum distance $d$. It has the abilities to detect up to $d-1$ quantum errors and correct up to $\lfloor\frac{d-1}{2}\rfloor$ quantum errors. A QECC must satisfy the quantum Singleton bound (i.e., $2d\leq n-k+2$) (see \cite{RefJ (2001) introduct 3}). In particular,
A QECC that satisfies $2d=n-k+2$ is called an MDS QECC.

Furthermore, we denote by $[[n,k,d;c]]_q$ a $q$-ary EAQECC
which encodes $k$ information qubits into
$n$ channel qubits with the help of $c$ pairs of maximally entangled states.
When $c=0$, the EAQECCs are just the standard quantum stabilizer codes.
Similarly, a EAQECC with parameters $[[n,k,d;c]]_q$ is required to satisfy
the following EA-quantum Singleton bound:

\begin{lemma}(EA-quantum Singleton bound)(\cite[Corollary 9]{RefJ (2022) Grassl bound})
For an any EAQECC with parameters $[[n,k,d;c]]_q$,
then we have
\begin{align}
&k\leq c+\max\{0,n-2d+2\};\label{eqqqq 1}\\
&k\leq n-d+1;\label{eqqqq 2}\\
&k\leq \frac{(n-d+1)(c+2d-2-n)}{3d-3-n},\ if\ 2d\geq n+2.\label{eqqqq 3}
\end{align}
An EAQECC that satisfies Eq. (\ref{eqqqq 1}) for $2d\leq n+2$ or Eq. (\ref{eqqqq 3}) for $2d\geq n+2$ with equality is called an MDS EAQECC.
\end{lemma}

The following lemma shows that from the known MDS codes with Galois hulls we can obtain two classes of EAQECCs.

\begin{lemma}\label{lem EAQECC}
If $\C$ is an $[n,k]_q$ MDS code, then
\begin{itemize}
\item[(1)]
(\cite[Corollary IV.1]{RefJ (2021) cao}) there exists an
$[[n,k-dim(\Hull_e(\C)),n-k+1;n-k-dim(\Hull_e(\C))]]_q$ EAQECC;
\item[(2)]
(\cite[Theorem 4.2]{RefJ (2023) Li.H DM})
there exists an
$[[n,n-k-dim(\Hull_e(\C)),k+1;k-dim(\Hull_e(\C))]]_q$ EAQECC.
\end{itemize}
\end{lemma}

The following lemma shows that we can obtain general entanglement-assisted concatenated quantum codes (EACQCs)
 from two EAQECCs.

\begin{lemma}\label{lem EACQC}
(\cite[Theorem 1]{RefJ (2022) EACQC})
Denote by $Q_1 :=[[n',k',d';c']]_p$ and $Q_2:=[[n,k,d;c]]_{p^{k'}}$ the inner code and the outer code, respectively. Then there exists an $[[n'n,k'k,\geq d'd; c'n+ck']]_p$ EACQC.
\end{lemma}

Combining Lemma \ref{lem EAQECC} and the MDS codes with Galois hulls obtained in this paper, we can obtain many EAQECCs.
Compared to the EAQECCs obtained from MDS codes with Galois hulls in \cite{RefJ (2021) cao,RefJ (2023) cao,RefJ (2022) Fang CC,RefJ (2023) Yang FFA,RefJ (2023) Islam,RefJ (2023) Li.H DM}, the EAQECCs we constructed have more flexible parameters.
Moreover, in Lemma \ref{lem EACQC}, if we take $k'=m$ and let $Q_2$ be the $q$-ary (MDS) EAQECCs
constructed in this paper, then we are able to obtain some families of $p$-ary
EACQCs with larger lengths.
 Since there are so many MDS codes with Galois hulls, (MDS) EAQECCs and EACQCs available, we will not list them all here.

\section{Conclusions}\label{sec6}

In this paper,
for a vector $\bm{v}\in (\F_q^*)^n$ and a $q$-ary linear code $\C$,
we give some necessary and sufficient conditions for the equivalent code $\Phi_{\bm{v}}(\C)$ of $\C$ and the extended code of $\Phi_{\bm{v}}(\C)$
to be $e$-Galois self-orthogonal (see Lemmas \ref{lem galois oa} and \ref{lem galois oa extend}).
This leads to some necessary and sufficient conditions for GRS and EGRS codes to be Galois self-orthogonal
(see Lemmas \ref{lem <a,v>}, \ref{lem <a,v> EG}, \ref{lem lam(x)} and \ref{lem lam(x) 2}).
Moreover, we construct many new classes of $e$-Galois self-orthogonal MDS codes that can have dimensions greater than $\lfloor \frac{n+p^e-1}{p^e+1}\rfloor$ (see Theorems \ref{th 1 ax+b}, \ref{th p2s 1} and \ref{th 11 3333}).
Galois self-orthogonal MDS codes of length $q+1$ are also obtained (see Theorem \ref{th 2 q+1}).
Furthermore, by propagation rules, we obtain some new MDS codes
with Galois hulls of arbitrary dimensions.
Then many quantum codes can be obtained from these MDS codes with Galois
hulls.

\section*{Acknowledgments}

This research is supported by the National Natural Science Foundation of China under Grant Nos. 12171134 and U21A20428.

\section*{Data availability}

All data generated or analyzed during this study are included in this paper.

\section*{Conflict of Interest}

The authors declare that there is no possible conflict of interest.

\section*{Appendix}

This appendix contains the omitted proofs for Lemma \ref{lem zhenchu} and Theorem \ref{th 11 3333}.

$\bullet$ \textbf{Proof of Lemma \ref{lem zhenchu}.}

\begin{proof}
It is easy to see that for any $0\leq i,j\leq k-1$,
we have $p^ei+j<p^{2e}-1$.
Suppose $(i,j)\neq (0,0)$.
Then if $p^ei+j=u(p^e-1)$, we have $1\leq u\leq p^e$.
We suppose that $p^ei+j$ is divisible by $p^e-1$ for some $0\leq i,j\leq k-1$.
Then there exists an integer $u$ such that
$$p^ei+j=u(p^e-1).$$
Hence, we have
$$p^ei+j=(u-1)p^e+p^e-u.$$
Since $1\leq u\leq p^e$,
it follows that
$$i=u-1,\quad j=p^e-u.$$

If $u\geq t+1$, then
$$i=u-1\geq t\geq k,$$
which contradicts the fact $i\leq k-1$;

If $u\leq p^e-t$, then
$$j=p^e-u\geq t\geq k,$$
which contradicts the fact $j\leq k-1$.
Hence, $u\in\{0\}\cup [p^e-t+1,t]$.
This completes the proof.
\end{proof}

$\bullet$ \textbf{Proof of Theorem \ref{th 11 3333}.}

\begin{proof}
Let $\theta=w^{\frac{q-1}{p^e-1}}$, where $w$ is the primitive element of $\F_{q}$.
Put
$$\bm{a}=(0, w, w\theta,\dots,w\theta^{p^e-2},\dots,w^r,w^{r}\theta,\dots,w^{r}\theta^{p^e-2})\in \F_{q}^{r(p^e-1)+1},$$
and
$$\bm{v}=(v_0,\underbrace{v_1,\dots,v_1}_{p^e-1\ times},\dots,\underbrace{v_r,\dots,v_r}_{p^e-1\ times})\in (\F_{q}^*)^{r(p^e-1)+1} ,$$
where $v_0,v_1,\dots,v_r\in \F_{q}^*$.
Then when $(i,j)=(0,0)$, we have
\begin{equation}\label{eq qiuhe 0}
\langle \bm{a}^{0},\bm{v}^{p^e+1}\rangle_E=v_0^{p^e+1}+(p^e-1)\sum_{l=1}^rv_l^{p^e+1}.
\end{equation}
Moreover, when $(i,j)\neq (0,0)$, we have
$$\langle \bm{a}^{p^ei+j},\bm{v}^{p^e+1}\rangle_E=\sum_{l=1}^rw^{l(p^ei+j)}v_l^{p^e+1}\sum_{\nu=0}^{p^e-2}\theta^{\nu(p^ei+j)},$$
thus
\begin{equation}\label{eq qiuhe 2}
\langle \bm{a}^{p^ei+j},\bm{v}^{p^e+1}\rangle_E=\begin{cases}
0, & if\ (p^e-1)\nmid (p^ei+j);\\
(p^e-1)\sum_{l=1}^rw^{l(p^ei+j)}v_l^{p^e+1}, & if\  (p^e-1)\mid (p^ei+j).
\end{cases}
\end{equation}
We only need to prove the case when $r$ is odd,
since the proof is completely similar for the case when $r$ is even.
Denote $t=\frac{p^e+r-1}{2}$.
By Lemma \ref{lem zhenchu}, $p^ei+j=u(p^e-1)$ if and only if $u\in [p^e-t+1,t]$.
Then by Eq. (\ref{eq qiuhe 2}), when $(p^e-1)\mid (p^ei+j)$,
we have
$$\langle \bm{a}^{p^ei+j},\bm{v}^{p^e+1}\rangle_E=(p^e-1)\sum_{l=1}^{r}w^{l u (p^e-1)}v_l^{p^e+1},$$
where  $u\in [p^e-t+1,t]$.
Let $\alpha=w^{p^e-1}$ and $a=p^e-t+1$, then we can get $\alpha^{a+v}\neq \alpha^{a+v'}\neq 1$
for any $0\leq v\neq v'\leq r-2$.
Let
$$B=\begin{pmatrix}
 1&   1    &  1  &  \dots  & 1  \\
 0& \alpha^{a} &  \alpha^{2a} & \dots & \alpha^{ra}\\
 0& \alpha^{a+1} &  \alpha^{2(a+1)}  &\dots & \alpha^{r(a+1)}\\
 \vdots& \vdots &  \vdots & \ddots & \vdots\\
 0& \alpha^{a+r-2} & \alpha^{2(a+r-2)} & \dots & \alpha^{r(a+r-2)}\\
\end{pmatrix}$$
be an $r\times (r+1)$ matrix over $\F_{q}$.
For $1\leq i\leq r+1$, let $B_i$ be the $r\times r$ matrix obtained from $B$ by deleting the $i$-th column.
It is easy to check that $rank(B)=rank(B_1)=\dots=rank(B_{r+1})=r$.
By \cite[Lemma 6]{RefJ (2023) Wan 2}, equation $B\bm{x}^T=\bm{0}^T$
has a solution $\bm{x}=(x_0,x_1,\dots,x_{r})\in (\F_{q}^*)^{r+1}$.
Since $\frac{m}{e}$ is odd and $p$ is even, by Lemma \ref{lem HHH}, we have $\F_q^*=\langle w^{p^e+1}\rangle$.
Then
let $v_l^{p^e+1}=x_l$ for $1\leq l\leq r$ and let $v_0\in \F_{q}^*$ such that $v_0^{p^e+1}=x_0(p^e-1)$.
By Eq. (\ref{eq qiuhe 0}), we have
$$\langle \bm{a}^{0},\bm{v}^{p^e+1}\rangle_E=v_0^{p^e+1}+(p^e-1)\sum_{l=1}^rv_l^{p^e+1}=(p^e-1)\sum_{l=0}^{r}x_l=0.$$
Moreover, by Eq. (\ref{eq qiuhe 2}), when $(i,j)\neq (0,0)$ and $(p^e-1)\mid (p^ei+j)$, we have
$$\langle \bm{a}^{p^ei+j},\bm{v}^{p^e+1}\rangle_E=(p^e-1)\sum_{l=1}^{r}\alpha^{lu}x_l=0,$$
where  $u\in [p^e-t+1,t]$.
Hence $\langle \bm{a}^{p^ei+j},\bm{v}^{p^e+1}\rangle_E=0$, for all $0\leq i,j\leq k-1$.
Then the desired result follows from Lemma \ref{lem <a,v>}.
This completes the proof.
\end{proof}

\end{document}